\documentclass[10pt,a4paper]{article}
\usepackage{epsfig}
\usepackage{graphicx}
\usepackage{amsfonts,amssymb}
\usepackage[standard]{ntheorem}

\begin{document}

\title{From Vicious Walkers to TASEP.}
\author{T.C. Dorlas$^{1}$, A.M. Povolotsky$^{1,2,}\thanks{%
Corresponding author e-mail address: alexander.povolotsky@gmail.com}$, V.B.
Priezzhev$^{2}$ \\
\\
$^{1}${\small {Dublin Institute for Advanced Studies,10 Burlington rd,
Dublin 4,Ireland } }\\
$^{2}${\small {Bogoliubov Laboratory for Theoretical Physics, Joint
Institute for } }\\
{\small {\ Nuclear Research,141980, Dubna, Russia} }}
\maketitle

\begin{abstract}
We propose a model of semi-vicious walkers, which interpolates between the
totally asymmetric simple exclusion process and the vicious walkers model,
having the two as limiting cases. For this model we calculate the
asymptotics of the survival probability for $m$ particles and obtain a
scaling function, which describes the transition from one limiting case to
another. Then, we use a fluctuation-dissipation relation allowing us to
reinterpret the result as the particle current generating function in the
totally asymmetric simple exclusion process. Thus we obtain the particle
current distribution asymptotically in the large time limit as the number of
particles is fixed. The results apply to the large deviation scale as well
as to the diffusive scale. In the latter we obtain a new universal
distribution, which has a skew non-Gaussian form. For $m$ particles its
asymptotic behavior \ is shown to be $e^{-\frac{y^{2}}{2m^{2}}}$ as $%
y\rightarrow -\infty $ and $e^{-\frac{y^{2}}{2m}}y^{-\frac{m(m-1)}{2}}$ as $%
y\rightarrow \infty $.
\end{abstract}


%


\section{Introduction}

Exact solutions of 1-dimensional (1D) many particle stochastic models \cite%
{Schutz review} have given much insight into the physics of non-equilibrium
systems in one dimension \cite{Spohn}. They serve as a testing ground for
the macroscopic theories, being able to verify their predictions \cite%
{Bertini}. Examples are the description of different kinds of
non-equilibrium phase transitions\cite{Evans}, calculation of the large
deviation functions for the density profile and total particle current\cite%
{Derrida}, verification of the fluctuation dissipation relations \cite%
{Lebowitz Spohn} and testing of the range of their validity \cite{Harris
Schutz}.

The range of models is very broad. In the context of the present article we
mention two of them. The first one is the lock step model of vicious walkers
(VW) that has been introduced in the physical literature by M. Fisher \cite%
{Fisher} to describe the wetting and melting phenomena. This is a random
process defined as many \textit{non-interacting} particles performing random
walks on a 1D lattice, whose space-time trajectories are forbidden to cross
each other. The term non-interacting means that the probability of a
particular realization of the process, which meets the latter constraint, is
\ given by the product of the probabilities of the random walks performed by
each individual walker. The other realizations, where crossings occur, are
assigned zero statistical weight. Such an elimination of a fraction of
possible outcomes at every time step violates the probability conservation.
A measure of the probability dissipation is the sum of the probabilities of
all possible particle configurations at a given time, referred to as the
survival probability. Its leading asymptotics for $m$ particles has been
shown by M. Fisher to decay with time $t$ as a power law: $t^{-\frac{m\left(
m-1\right) }{4}}$.

Another model, the totally asymmetric simple exclusion process (TASEP) \cite%
{Ligget}, has been widely discussed in connection to the Kardar-Parisi-Zhang
universality class \cite{KPZ}. In contrast to the VW model, this is a model
of \textit{interacting} random walks. The interaction prevents particles
from jumping to occupied sites. Therefore, similarly to vicious walkers, the
statistical ensemble includes only those events in which the space-time
trajectories of particles do not cross. The difference is that there is an
interaction that changes the statistical weights of particle trajectories
when they pass via neighboring sites so that the total probability is
conserved. In this case the quantity corresponding to the survival
probability is just a probability normalization constant.

Thus, the probability lost after imposing the global non-crossing constraint
on the dynamics of non-interacting particles in VW is regained in the TASEP
by adding the lacking probability locally at certain steps. In the present
paper we consider the two models as limiting cases of a more general
interaction, where the added probability is a varying parameter of the model
that controls the probability dissipation, such that the probability
conservation is restored when it is tuned to the TASEP value. In this
connection a natural question arises: what happens with Fisher's asymptotics
for the survival probability under such a generalization and, in particular,
how does it cross over to the TASEP normalization constant. This is the
first question we address in this paper.

Specifically, we propose a semi-vicious walkers (SVW) model, which
interpolates between VW and TASEP. It is a model of interacting particles
with partial repulsion or attraction, where trajectory crossings are
forbidden. The term partial repulsion (attraction) means that the
probability for the particle to jump to an occupied site is not equal to
zero like in TASEP but can be less (greater) than that of a free particle.
At the same time, the non-crossing constraint leads to lack of probability
conservation in the same way as in the VW model. The strength of the
interaction, which also characterizes the probability dissipation, is a
parameter of the model, which has the TASEP and VW as limiting cases at the
endpoints of its range.

In this article we obtain the large-time asymptotics of the survival
probability. Its limiting case corresponding to VW is given by the above
mentioned result of Fisher, which yields the leading power law asymptotics.
Later it was reproduced with more mathematical rigor together with the
constant prefactor that was obtained for the particular initial
configurations, where the particles are separated by equal spaces \cite%
{Krattenhaler Guttman Viennot}. In the case of a general initial
configuration of walkers this prefactor depends on the initial positions.
This case has been studied in \cite{Rubey}.

Our results can be roughly divided into two parts. For generic values of the
interaction strength, away from the point corresponding to the TASEP, the
probability dissipation is finite. It is intuitively clear that the main
asymptotics must be similar to the VW one. Indeed, we obtain the Fisher's
power law with a constant prefactor that depends on the initial positions of
the particles and on the interaction strength. It is shown to diverge in the
TASEP limit. The second and probably the most interesting case is the
transition region, which interpolates between the VW and TASEP behavior. To
probe into this region, we consider a scaling limit of the survival
probability, where the large time limit and the TASEP limit of the
interaction strength are combined. In this way we obtain a scaling function
of a single parameter that controls the transition from VW to TASEP.

The second problem we address is the distribution of the integral particle
current in TASEP. A first example of such an exact distribution has been
obtained by Derrida and Lebowitz \cite{Derrida Lebowitz}, who found the
large deviation function for the particle current in the TASEP confined to a
ring. A specific property of the finite system is that there is a finite
relaxation time, after which the system settles into a non-equilibrium
stationary state, independent of initial conditions \cite{Gwa Spohn}. Then,
the tool used to study integral current fluctuations is, roughly speaking,
an analysis of the relaxation of the system subject to a perturbation into
the stationary state. Technically it is an analysis of the largest
eigenvalue of the perturbed Markov matrix governing the process.

In genuinely infinite systems the situation is more peculiar. In this case
there is no characteristic relaxation time scale. When starting away from
the stationary state, the latter is never approached. In this case one needs
to consider actual time evolution of quantities of interest, which depend on
initial conditions. A major breakthrough in this direction has been achieved
by Johansson, \cite{Johansson}. He considered the TASEP evolution of an
infinite cluster of particles, which initially occupies all sites of the
lattice to the left of a fixed site, and calculated the distribution of the
number of steps made by an arbitrary particle in this cluster. Johansson's
solution has initiated a burst of activity in the field, which exploited
deep connections of the TASEP to the theory of random matrix ensembles and
the determinantal point processes. Results have been obtained for different
initial conditions and extended to many particle joint distributions \cite%
{Sasamoto Nagao},\cite{Rakos Schuetz},\cite{Sasamoto},\cite{Ferrari Spohn},%
\cite{Borodin Ferrari},. Remarkably, in the scaling limit these results
provide parameter free universal distributions \cite{Prahofer Spohn} of the
fluctuations measured in the KPZ characteristic scale, which is of order of $%
t^{1/3}$ as time $t$ grows to infinity \cite{Krug Spohn}. This is in
contrast to the diffusive scale $t^{1/2}$, which, according to the Central
Limit Theorem (CLT), characterizes the fluctuations of the distance
travelled by a free particle \cite{Feller}. The large deviation limit of the
single particle current distribution has been studied in \cite{Harris Rakos
Schuetz} in connection with the fluctuation dissipation relations.

Despite the great success in finding the distributions of single particle
currents and their correlation functions, very few results on the integral
particle current, i.e. on the distance travelled by all particles, are
available for driven diffusive systems. In fact the only known exact result
is the above mentioned large deviation function for the integral particle
current for the TASEP in a ring \cite{Derrida Lebowitz} and its
generalization for the partially asymmetric case \cite{Kim}, \cite{Lee Kim}.
No results beyond the large deviation scale, neither a generalization for an
infinite system has been proposed. On the other hand, extensive quantities
like the integral current, are important ingredients of the thermodynamics
of the models. A knowledge of the character of their fluctuations could be
of help for extension of the thermodynamical formalism to irreversible
systems. The present paper makes a step in this direction. The problem we
solve here is as follows. We study the large time asymptotics of the
distribution of the total number of jumps made by a finite number of TASEP
particles in an infinite lattice, given an arbitrary initial configuration.
The idea that allows us to consider this problem in line with the previous
one is the existence of a kind of fluctuation-dissipation relation that
unifies the dissipation of probability in SVW and the statistics of
fluctuations of the integrated particle current in TASEP. Specifically, an
auxiliary parameter, which violates the probability conservation, can be
introduced into the evolution operator in TASEP to account for the total
number of steps made by particles, see e.g. \cite{Derrida Lebowitz}. This
parameter plays a role similar to the one played by the interaction in SVW,
the two problems being equivalent after a certain change of variables. Using
this fact, we interpret the result obtained for the survival probability in
SVW as a generating function of the particle current in TASEP. The latter,
in its turn, can be used to reconstruct the form of the current distribution.

Like those for SVW, the results obtained for the TASEP particle current
consist of two parts. The generic values of the interaction strength
correspond to the distribution of the particle current at the large
deviation scale, i.e. describes the deviations of order of time $t$. It
turns out that it has a skew distribution with asymmetric negative and
positive tails. These tails are connected by a middle part corresponding to
the transition region. The latter yields the current distribution at the
diffusive scale, $t^{1/2}$, which is shown to have a skew non-Gaussian form,
depending only on the total number of particles, and we suggest to be
universal for particles performing a driven diffusion.

One technical remark has to be made about the connection of our solution to
the theory of random matrix ensembles. It is this connection which \ enabled
the above mentioned progress in calculating the single particle current
distributions and their many particle generalizations. In our solution this
connection has also been exploited. Namely, the survival probability in the
SVW model at generic values of the interaction strength and exactly at the
TASEP point can be calculated in terms the Mehta integrals $I_{m,k}$ with $%
k=1/2$ and $k=1$, which appear as normalization factors in the orthogonal
and unitary Gaussian ensembles of random matrices respectively \cite{Mehta
Dyson}. Note, however, that the scaling function obtained in the transition
region for the system of $m$ particles can be reduced to neither of these
integrals except of at three limiting points, where it becomes $I_{m,1/2}$,$%
I_{m,1}$ and $I_{m-1,1}$ respectively. Thus, we obtain a generalized object,
which interpolates between these three Mehta integrals, and, therefore, in a
sense unifies three different matrix ensembles. To our knowledge no such
generalization has appeared in the theory before.

The article is organized as follows. In Section 2 we formulate the SVW
model, state the results obtained and discuss their interpretation in terms
of the probability distribution of the particle current in TASEP. Sections
3-5 are a technical part, where we prove the results outlined in Section 2.
In Section 3 we solve the master equation for the SVW model. In Section 4 we
obtain the asymptotic formulas for the transition probabilities. In Section
5 we prove the limiting properties of the function characterizing the SVW to
TASEP transition. Section 6 has a summary and conclusions.

\section{Model and results}

\subsection{Semi vicious walkers model}

Consider $m$ particles on a 1D infinite lattice. A configuration $X$ of the
system is specified by an $m$-tuple of strictly increasing integers%
\begin{equation}
X=\left\{ x_{1}<x_{2}<\cdots <x_{m}\right\} ,
\end{equation}%
where $x_{i}$ is the coordinate of $i$-th particle. The strictly increasing
order reflects the exclusion condition, i.e. two particles cannot occupy the
same site. We say that the relation \thinspace $X\leq Y$ \ holds for
particle configurations if%
\begin{equation}
x_{1}\leq y_{1}\leq x_{2}\leq \cdots \leq x_{m}\leq y_{m}.
\end{equation}%
The SVW model is a random process, which is defined on a set of sequences of
configurations $X^{0},X^{1},\cdots ,X^{t},$ such that
\begin{equation}
X^{0}\leq X^{1}\cdots \leq X^{t}.
\end{equation}%
We refer to such a sequence as a trajectory of the system traveled up to
time $t$. Every such trajectory is realized with probability
\begin{equation}
P(X^{0},\ldots ,X^{t})=T(X^{t},X^{t-1})\cdots
T(X^{2},X^{1})T(X^{1},X^{0})P_{0}(X^{0}).
\end{equation}%
$P_{0}(X)$ is the initial probability of the configuration $X$ and the
transition probability $T(X,Y),$ from the configuration $Y$ to $X,$ is
defined as follows%
\begin{equation}
T(X,Y)=\vartheta (x_{m}-y_{m})\prod\limits_{k=1}^{m-1}\theta \left(
x_{i}-y_{i},x_{i+1}-y_{i}\right) ,  \label{weights}
\end{equation}%
where
\begin{eqnarray}
\vartheta (k) &=&\left( 1-p\right) \delta _{k,0}+p\delta _{k,1},
\label{weights_11} \\
\theta \left( k,l\right) &=&\left( 1-p\left( 1-\kappa \delta _{l,1}\right)
\right) \delta _{k,0}+p\delta _{k,1},  \label{weights_12}
\end{eqnarray}%
and%
\begin{eqnarray}
0 &<&p<1,  \label{p range} \\
1-1/p &\leq &\kappa \leq 1.  \label{kappa range}
\end{eqnarray}%
This means that at each discrete time step a particle can jump forward with
probability $p$ or stay put with probability $1-p$, provided that the next
site is empty. If the next site is occupied, the probability for a particle
to stay put is $\left( 1-p(1-\kappa )\right) $. The probability deficit $%
p(1-\kappa )$, corresponds to the process when the particle jumps to the
adjoining occupied site, which is forbidden. This excluded process results
in probability dissipation in this model. The form of the transition
probabilities corresponds to the backward sequential update, i.e. the
particles are updated starting from the $m$-th particles one by one in
backward direction. In particular limiting cases the model reduces to

\begin{enumerate}
\item $\kappa =0$ \ - \ VW, a particle jumps forward with probability $p$ or
stays with probability $\left( 1-p\right) $, irrespective of whether the
next site is occupied or not. But then those realizations of the process
where two particles come to the same site must be removed from the
statistical ensemble.

\item $\kappa =1$ \ - \ TASEP, a particle jumps forward with probability $p$
or stays with probability $\left( 1-p\right) $ provided the adjoining site
is empty. When the next site is occupied the particle stays where it is with
probability $1$.
\end{enumerate}

The TASEP with the backward sequential update was studied in \cite{Brankov
Priezzhev Shelest} and \cite{Rakos Schuetz}, where it was referred to as a
fragmentation model. In the case $\kappa =\left( 1-1/p\right) $, the
probability for a particle to stay where it is when the next site occupied,
is zero. Therefore, the trajectories of particles passing via neighboring
sites have zero weight, i.e. they are removed from the ensemble as well as
those which meet at the same site. Therefore this situation resembles the
vicious walks of dimers. The range of $\kappa $ given in (\ref{kappa range})
is due to the requirement for $(1-p(1-\kappa ))$ to be a probability.
Positive values of $\kappa $ correspond to repulsive interaction, while
negative values correspond to an attractive interaction. The domain $\kappa
>1$ is also of interest in connection with the current fluctuation in TASEP,
though it does not have a probabilistic meaning in the context of SVW.

\subsection{The results about the SVW model}

Let us consider the quantity%
\begin{equation}
\mathcal{P}_{t}\left( X^{0}\right) =\sum\limits_{X^{0}\leq X^{1}\cdots \leq
X^{t}}P(X^{0},\ldots ,X^{t}),
\end{equation}%
where the sum is over all the trajectories of the system starting at the
configuration $X^{0}$, i.e. $P_{0}\left( X\right) =\delta _{X,X^{0}}$. This
quantity is the partition function of the statistical ensemble of the
trajectories with the statistical weights defined above. On the other hand,
if we add the lacking processes allowing the particles to jump to an
occupied site, the value of $\mathcal{P}_{t}\left( X^{0}\right) $ will have
the meaning of probability for all the particles not to meet up to time $t$.
In Fisher's original formulation of such a process, two particles annihilate
when getting to the same site. Then, $\mathcal{P}_{t}\left( X^{0}\right) $
is the probability for $m$ particles to survive up to time $t$. Therefore,
we refer to this quantity as a survival probability. Below we formulate
three theorems, which specify the asymptotic behaviour of the survival
probability in the limit of large time for different parts of the range of
the parameter $\kappa $. The proof of these theorems is the content of
Sections 4,5.

\begin{remark}
\label{complex kappa remark} Two of the theorems below are stated and proved
for complex valued parameter $\kappa $ and the third one for real $\kappa >1$%
. Obviously, the quantity obtained there has a meaning of the survival
probability only for real $\kappa $ varying in the range (\ref{kappa range}%
). Consideration of other values of $\kappa $ is justified by its later
interpretation in terms of the generating function of the moments of the
total particle current in the TASEP. In the latter case the complex values
of $\kappa $ turn out to be meaningful and useful to reconstruct the total
particle current distribution in the TASEP.
\end{remark}

\subsubsection{The survival probability for SVW}

\paragraph{\ Generic case, $\left\vert \protect\kappa \right\vert <1$}

In this case the asymptotic behavior of the survival probability $\mathcal{P}%
_{t}\left( X^{0}\right) $ as $t\rightarrow \infty $ is given by the
following theorem.

\begin{theorem}
\label{kappa<1 Theorem}Let $\kappa \in \mathbb{C}$ be a fixed complex number
from the domain $\left\vert \kappa \right\vert <1$, and let $\left\vert
x_{i}^{0}-x_{j}^{0}\right\vert <\infty $ for any $i,j=1,\ldots ,m$. Then, as
$t\rightarrow \infty $ the survival probability $\mathcal{P}_{t}\left(
X^{0}\right) $ for $m$ particles is
\begin{equation}
\mathcal{P}_{t}\left( X^{0}\right) =A\left( \kappa ;X^{0}\right) \left[
tp(1-p)\right] ^{-\frac{m\left( m-1\right) }{4}}\left[ 1+O\left( \left( \log
t\right) ^{3}t^{-1/2}\right) \right] ,  \label{P_t(X^0) kappa<1}
\end{equation}%
where the prefactor is given by%
\begin{equation}
A(\kappa ;X^{0})=\frac{2^{m}\prod_{l=1}^{m}\Gamma \left( l/2+1\right) }{%
\left( 1-\kappa \right) ^{\frac{m\left( m-1\right) }{2}}\pi ^{m/2}}\det %
\left[ g_{i,j}(x_{m}^{0}-x_{i}^{0};\kappa )\right] _{_{1\leq i,j\leq m}}
\label{A(X^0,kappa)}
\end{equation}%
and where the function $g_{i,j}(x;\kappa )$ is defined by
\begin{equation}
g_{i,j}(x;\kappa )=\oint_{C_{0}}\frac{d\xi }{2\pi \mathrm{i}}\frac{\left(
\kappa +\kappa \xi -1\right) ^{i-1}\left( 1+\xi \right) ^{x}}{\xi ^{j}}.
\label{g_i,j}
\end{equation}
\end{theorem}

Thus, in the range $\kappa <1$, up to the factor $A(X^{0};\kappa )$, which
captures the dependence on the initial configuration $X^{0}$, the survival
probability reproduces Fisher's power law. All the dependence on $X^{0}$ is
in fact hidden in the determinantal part of $A(X^{0};\kappa )$. In some
particular cases the determinant can be simplified to a more transparent
expression. For example, for equidistant initial conditions,%
\begin{equation}
x_{m}^{0}-x_{i}^{0}=a\left( m-i\right) ,
\end{equation}%
where $a$ is a positive integer, it can be calculated explicitly:%
\begin{equation}
\det \left[ g_{i,j}(x_{m}^{0}-x_{i}^{0};\kappa )\right] _{_{1\leq i,j\leq
m}}=(a+\kappa -a\kappa )^{\frac{m\left( m-1\right) }{2}}.
\end{equation}

In the limit $\kappa \rightarrow 0$ the determinant reduces to
\begin{equation}
\det [g_{i,j}(x_{m}^{0}-x_{i}^{0};0)]=\prod\limits_{1\leq i<j\leq m}\frac{%
x_{j}^{0}-x_{i}^{0}}{j-i}.
\end{equation}%
Then, up to rescaling of space and time, one recovers the result \cite{Rubey}
for VW:
\begin{equation}
A(0;X^{0})=\prod\limits_{1\leq i<j\leq m}\left( x_{j}^{0}-x_{i}^{0}\right)
\left\{
\begin{array}{ll}
\pi ^{-\frac{m}{4}}2^{-\frac{m\left( m-2\right) }{4}}\prod\limits_{l=1}^{%
\frac{m}{2}}\frac{1}{\left( 2l-1\right) !}; & \mathrm{even}m\  \\
\pi ^{\frac{1}{4}-\frac{m}{4}}2^{-\frac{\left( m-1\right) ^{2}}{2}%
}\prod\limits_{l=1}^{\frac{\left( m-1\right) }{2}}\frac{1}{\left( 2l\right) !%
}; & \mathrm{odd}m\
\end{array}%
\right. .  \label{A(X^0;0)}
\end{equation}%
In the limit $\kappa \rightarrow 1$, we have
\begin{equation}
\det \left[ g_{i,j}(x_{m}^{0}-x_{i}^{0};1)\right] _{_{1\leq i,j\leq m}}=1.
\label{g_ij kappa->1}
\end{equation}%
This limit corresponds to the TASEP. Hence the asymptotics must change, as
the probability conservation is restored. The signature of this fact is the
divergence of the term $A(\kappa ;X^{0})$ that takes place in this limit.
Specifically, (\ref{A(X^0,kappa)}) and (\ref{g_ij kappa->1}) suggests that
as $\kappa $ approaches one, $A\left( \kappa ;X^{0}\right) $ diverges as $%
\left( 1-\kappa \right) ^{-m\left( m-1\right) /2}$. Comparing the exponent
of this expression with the one of the time decay $t^{-m\left( m-1\right)
/4} $, we can guess that the transition takes place on the scale $\left(
1-\kappa \right) \sim 1/\sqrt{t}$. This hypothesis is justified below.

\paragraph{Generic case, $\protect\kappa >1$}

In this case no values of $\kappa $ fall into the range (\ref{kappa range}).
Therefore, according to the Remark \ref{complex kappa remark}, the result
formally obtained for $\mathcal{P}_{t}\left( X^{0}\right) $ does not have a
probabilistic meaning in terms of SVW. However, it is still meaningful for
the description of current fluctuation in the TASEP.

\begin{theorem}
\label{kappa>1 Theorem}Let $\kappa \in \left( 1,\infty \right) $ be a fixed
real number, and let $\left\vert x_{i}^{0}-x_{j}^{0}\right\vert <\infty $
for any $i,j=1,\ldots ,m$. Then, as $t\rightarrow \infty $, the survival
probability $\mathcal{P}_{t}\left( X^{0}\right) $ for $m$ particles is
\begin{eqnarray}
\mathcal{P}_{t}\left( X^{0}\right) &=&\left( 1-p+\kappa p\right) ^{\left(
m-1\right) t}\left( 1-p+\kappa ^{1-m}p\right) ^{t}  \label{P_t(X^0) kappa>1}
\\
&\times &\frac{\left[ m\right] _{\kappa }^{m-1}\kappa ^{x_{1}^{0}+\cdots
+x_{m-1}^{0}-(m-1)x_{m}^{0}}}{\left[ m-1\right] _{\kappa }!}\left[ 1+O\left(
\left( \log t\right) ^{3}t^{-1/2}\right) \right]  \nonumber
\end{eqnarray}
\end{theorem}

Here we use the common notations
\begin{equation}
\left[ m\right] _{q}=\frac{1-q^{m}}{1-q}
\end{equation}%
for the $q$-number and
\begin{equation}
\left[ m\right] _{q}!=\left[ 1\right] _{q}\cdots \left[ m\right] _{q}
\end{equation}%
for the $q$-factorial. Note that the $q$-numbers turn to usual numbers $%
\left[ m\right] _{q}\rightarrow m$ in the limit $q\rightarrow 1$.

\paragraph{\protect\bigskip Transition regime $\protect\kappa \rightarrow 1$}

Consider the limit%
\begin{equation}
t\rightarrow \infty ,\kappa \rightarrow 1,\left( 1-\kappa \right) \sqrt{t}%
=const.  \label{transreg}
\end{equation}%
We introduce the parameter%
\begin{equation}
\alpha =\lim_{t\rightarrow \infty }\left[ \left( 1-\kappa \right) \sqrt{%
tp\left( 1-p\right) }\right] .  \label{alpha}
\end{equation}

\begin{theorem}
\label{kappa->1 theorem}Let the condition $\left\vert x_{i}^{0}-x_{j}^{0}\right\vert
<\infty $ hold for all $i,j=1,\ldots ,m$. Then,in the limit (\ref{transreg}%
), for the parameter $\alpha \in \mathbb{C}$ defined in (\ref{alpha}) taking
any fixed complex value, $\mathcal{P}_{t}\left( X^{0}\right) $ converges to
\begin{equation}
\mathcal{P}_{t}=f_{m}\left( \alpha \right) \left[ 1+O\left( \frac{\left(
\log t\right) ^{3}}{\sqrt{t}}\right) \right] ,  \label{result transition}
\end{equation}%
where the function\ $f_{m}\left( \alpha \right) $ has the form of a multiple
integral:
\begin{eqnarray}
f_{m}\left( \alpha \right) &=&\frac{\left( -1\right) ^{\frac{m\left(
m-1\right) }{2}}}{\left( 2\pi \right) ^{\frac{m}{2}}2!\cdots \left(
m-2\right) !}  \nonumber \\
&\times &\int\limits_{-\infty }^{\infty }du_{1}\int\limits_{u_{1}}^{\infty
}du_{2}\cdots \!\!\int\limits_{u_{m-1}}^{\infty
}du_{m}\int\limits_{0}^{\infty }d\nu _{2}\cdots \!\int\limits_{0}^{\infty
}d\nu _{m}  \label{f_m(alpha)} \\
&\times &e^{-\frac{1}{2}u_{1}^{2}}\prod\limits_{i=2}^{m}\nu _{i}^{i-2}e^{-%
\frac{1}{2}\left( u_{i}+\nu _{i}\right) ^{2}-\alpha \nu _{i}}\Delta \left(
u_{1},\nu _{2}+u_{2},\ldots ,\nu _{m}+u_{m}\right) ,  \nonumber
\end{eqnarray}%
where
\begin{equation}
\Delta \left( x_{1},\ldots ,x_{m}\right) =\prod\limits_{1\leq i<j\leq
m}\left( x_{i}-x_{j}\right)
\end{equation}%
is the Vandermonde determinant.
\end{theorem}

The argument of $\mathcal{P}_{t}$ in (\ref{result transition}) can be
omitted as the dependence on the initial configuration is lost in the limit
under consideration. The limiting behaviors of $f_{m}(\alpha )$ match the
TASEP and VW asymptotics. Indeed, we prove in Section~5 that
\begin{eqnarray}
&&\lim_{\Re\alpha \rightarrow \infty }\alpha ^{\frac{m\left( m-1\right) }{2}%
}f_{m}\left( \alpha \right) =\frac{1}{m!}\frac{2^{m}}{\pi ^{m/2}}%
\prod\limits_{l=1}^{m}\Gamma \left( l/2+1\right) ,  \label{alpha->+infty} \\
&&f_{m}\left( 0\right) =1  \label{alpha=0} \\
&&\lim_{\Re\alpha \rightarrow -\infty }e^{-\alpha ^{2}\frac{m(m-1)}{2}%
}f_{m}\left( \alpha \right) =\frac{m^{m-1}}{\left( m-1\right) !}
\label{alpha->-infty}
\end{eqnarray}%
In (\ref{alpha->+infty}) and (\ref{alpha->-infty}) the imaginary part of $%
\alpha $ is implied to take an arbitrary fixed value. The proof of these
three limits given in Section \ref{f_m(alpha)_section} is done by reducing $%
f_{m}\left( \alpha \right) $ to different cases of Mehta integrals, $%
I_{m,1/2}$, $I_{m.,1}$ and $I_{m-1,1}$ respectively, which are well known in
the theory of Gaussian random matrix ensembles. \cite{Mehta Dyson}.

The above results can be illustrated by the example of the two particle
case, $m=2$, when the integral (\ref{f_m(alpha)}) for $f_{2}\left( \alpha
\right) $ simplifies significantly.
\begin{equation}
f_{2}\left( \alpha \right) =e^{\alpha ^{2}}\mathrm{Erfc}\left( \alpha
\right) .  \label{f_2(alpha)}
\end{equation}%
Here $\mathrm{Erfc}\left( \alpha \right) $ is the complementary Error
function%
\begin{equation}
\mathrm{Erfc}\left( \alpha \right) =\frac{2}{\sqrt{\pi }}\int\limits_{\alpha
}^{\infty }dxe^{-x^{2}}.
\end{equation}%
In Fig.~\ref{f_2(alpha)_fig} we show how $f_{2}\left( \alpha \right) $
interpolates between the SVW and TASEP limiting cases, which are (\ref%
{alpha->+infty}) and (\ref{alpha=0}) respectively.
\begin{figure}[tbp]
\unitlength=1mm \makebox(110,50)[cc] {\psfig{file=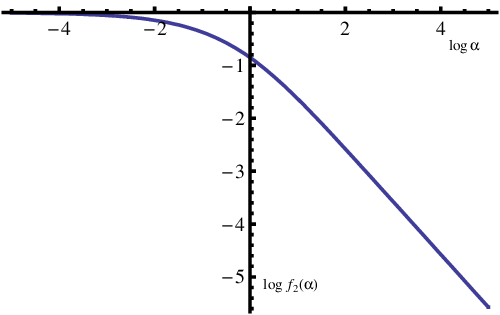,width=80mm}}
\caption{The log-log plot of the function $f_{2}(\protect\alpha )$ in the
range $\protect\alpha >0$. }
\label{f_2(alpha)_fig}
\end{figure}

\subsection{Current fluctuations in TASEP}

Consider the process with the transition weights $\widetilde{T}(X,Y)$
defined similarly to (\ref{weights}) but where the functions $\vartheta (k)$
and $\theta \left( k,l\right) $ are replaced by
\begin{eqnarray}
\widetilde{\vartheta }(k) &=&\left( 1-\widetilde{p}\right) \delta
_{k,0}+e^{\gamma }\widetilde{p}\delta _{k,1},  \label{weights_21} \\
\widetilde{\theta }\left( k,l\right) &=&\left( 1-\widetilde{p}\left(
1-\delta _{l,1}\right) \right) \delta _{k,0}+e^{\gamma }\widetilde{p}\delta
_{k,1}.  \label{weights_22}
\end{eqnarray}%
Here $0<\widetilde{p}<1$ and $\gamma $ is a complex-valued parameter. It is
not difficult to see that these transition weights correspond to the TASEP
except that for each particle, the probability $\widetilde{p}$ to jump is
multiplied by an additional factor $e^{\gamma }$, i.e.%
\begin{equation}
e^{\gamma Y_{t}}P_{TASEP}(X^{0},\ldots ,X^{t})=\widetilde{T}%
(X^{t},X^{t-1})\cdots \widetilde{T}(X^{1},X^{0})P_{0}(X_{0}),
\end{equation}%
where $P_{TASEP}(X^{0},\ldots ,X^{t})$ is the probability for a sequence of
particle configurations $X^{0},\ldots ,X^{t}$, to occur in the TASEP for $t$
successive steps and $Y_{t}$ is the total number of jumps made by all
particles in this sequence of configurations. Thus, one can calculate the
moment generating function for the cumulative particle current as follows,
\begin{equation}
\left\langle e^{\gamma Y_{t}}\right\rangle _{TASEP}=\sum\limits_{X^{0}\leq
X^{1}\cdots \leq X^{t}}e^{\gamma Y_{t}}P_{TASEP}(X^{0},\ldots ,X^{t}).
\label{generating}
\end{equation}%
On the other hand, we can see that, if we define%
\begin{eqnarray}
\kappa &=&e^{-\gamma },  \label{tasep-svw_1} \\
p &=&\frac{\widetilde{p}}{\left( 1-\widetilde{p}\right) e^{-\gamma }+%
\widetilde{p}},  \label{tasep-svw_2}
\end{eqnarray}%
then the following relation exists between the transition weights $%
\widetilde{T}(X,X^{\prime })$ defined in (\ref{weights_21}),(\ref{weights_22}%
) and those of SVW, (\ref{weights_11}),(\ref{weights_12}),
\begin{equation}
\left( 1-\widetilde{p}\right) ^{-m}\widetilde{T}(X,X^{\prime })=\left(
1-p\right) ^{-m}T(X,X^{\prime }).
\end{equation}%
As a result we have
\begin{equation}
\left\langle e^{\gamma Y_{t}}\right\rangle _{TASEP}=\left( 1+\widetilde{p}%
\left( e^{\gamma }-1\right) \right) ^{mt}\mathcal{P}_{t}\left( X^{0}\right) ,
\label{cumgen}
\end{equation}%
where $\mathcal{P}_{t}\left( X^{0}\right) $ is the survival probability
calculated for the SVW model, and the parameters $\kappa $ and $p$ of SVW
are related to the parameters $\widetilde{p},\gamma $ of TASEP by (\ref%
{tasep-svw_1}), (\ref{tasep-svw_2}). The function (\ref{cumgen}) encodes all
information about the distribution of the integrated particle current. Thus,
we can apply the Theorems \ref{kappa<1 Theorem}-\ref{kappa->1 theorem} to
obtain the asymptotic form of this distribution.

\paragraph{Large deviation function.}

It follows from the Theorems \ref{kappa<1 Theorem} and \ref{kappa>1 Theorem}
and the formula (\ref{cumgen}) that for fixed $\gamma \in \mathbb{R}$ the
asymptotic form of the generating function of the particle current $Y_{t}$
is as follows
\begin{equation}
\left\langle e^{\gamma Y_{t}}\right\rangle _{TASEP}\simeq \left\{
\begin{array}{cc}
\frac{\left( 1+\widetilde{p}\left( e^{\gamma }-1\right) \right) ^{mt+\frac{%
m\left( m-1\right) }{2}}}{\left[ te^{\gamma }\widetilde{p}(1-\widetilde{p})%
\right] ^{\frac{m\left( m-1\right) }{4}}}A\left( e^{-\gamma },X^{0}\right) &
\gamma >0 \\
\left( 1-\widetilde{p}\left( 1-e^{m\gamma }\right) \right) ^{t}\frac{\left(
1-e^{-\gamma m}\right) ^{m-1} e^{\gamma \sum_{i=1}^{m-1}(x_m^0-x_i^0) }}{%
\left( 1-e^{-\gamma }\right) \cdots \left( 1-e^{-\gamma \left( m-1\right)
}\right) } & \gamma \leq 0%
\end{array}%
\right. .
\end{equation}%
From here we conclude that a scaled cumulant generating function of the
random variable $Y_{t}/t$ exists%
\[
\lambda \left( \gamma \right) \equiv \lim_{t\rightarrow \infty }\frac{1}{t}%
\log \left\langle e^{Y_{t}\gamma }\right\rangle _{TASEP}=\left\{
\begin{array}{cc}
m\log \left( 1+\widetilde{p}\left( e^{\gamma }-1\right) \right) & \gamma
\geq 0 \\
\log \left( 1-\widetilde{p}\left( 1-e^{m\gamma }\right) \right) & \gamma
\leq 0%
\end{array}%
\right. .
\]%
It is convex and differentiable everywhere. Therefore, we refer to the G\"{a}%
rtner-Ellis theorem \cite{Gartner},\cite{Ellis} to show that the random
variable $v_{t}$ satisfies the large deviation principle with a rate function%
\[
I\left( v\right) \equiv \lim_{t\rightarrow \infty }\frac{1}{t}\log P\left(
Y_{t}/t=v\right) =\sup_{\gamma }\left( \gamma v-\lambda \left( \gamma
\right) \right) .
\]%
The solution of the maximization problem yields

\[
I\left( v\right) =\left\{
\begin{array}{cc}
mB\left( v/m\right) & v\geq m\widetilde{p} \\
B\left( v/m\right) & v\leq m\widetilde{p}%
\end{array}%
\right. ,
\]%
where $B\left( v\right) $ is the usual rate function of the Bernoulli process%
\[
B\left( v\right) =\left( 1-v\right) \log \frac{1-p}{1-v}+v\log \frac{p}{v}.
\]

\paragraph{Central limit theorem scaling.}

The transition regime corresponds to the following scaling limit
\begin{equation}
t\rightarrow \infty ,\qquad \gamma \rightarrow 0,\qquad \gamma \sqrt{t}%
=const.  \label{scaling_lim}
\end{equation}%
To translate the results obtained for this case we consider the random
variable%
\begin{equation}
y=\lim_{t\rightarrow \infty }\frac{Y_{t}-m\widetilde{p}t}{\sqrt{t\widetilde{p%
}\left( 1-\widetilde{p}\right) }},
\end{equation}%
where the convergence is in distribution. Taking the limit (\ref{scaling_lim}%
) in (\ref{cumgen}), using the Theorem \ref{kappa->1 theorem} and noting
that $\widetilde{p}\rightarrow p$ as $\gamma \rightarrow 0$ we obtain
\begin{equation}
\left\langle e^{\alpha y}\right\rangle _{TASEP}=e^{m\alpha
^{2}/2}f_{m}\left( \alpha \right) ,  \label{generating_transition}
\end{equation}%
where $\alpha $ is an arbitrary complex valued parameter related to $\gamma $
via -%
\begin{equation}
\alpha =\lim_{t\rightarrow \infty }\gamma \sqrt{tp\left( 1-p\right) }.
\end{equation}%
The random variable $y$ is the rescaled deviation of the integrated current $%
Y_{t}$ from $m\widetilde{p}t$, i.e. from the average value of $Y_{t}$ for $m$
non-interacting particles jumping with probability $\widetilde{p}$. Note
that $y$ is the variable that, in the case of free non-interacting
particles, satisfies the conditions for the applicability of the Central
Limit Theorem (CLT). According to the CLT the probability density function
(PDF) of $y$ for $m$ independent particles is the Gauss distribution,
\begin{equation}
P_{m}^{\mathrm{free}}(y)=\exp \left( -y^{2}/\left( 2m\right) \right) /\sqrt{%
2\pi m}
\end{equation}%
Correspondingly, the generating function of its moments is
\begin{equation}
\left\langle e^{\alpha y}\right\rangle _{\mathrm{free}}=\exp \left( m\alpha
^{2}/2\right) ,
\end{equation}%
which is the first factor in the moment generating function (\ref%
{generating_transition}). Therefore, the form of the second factor, $%
f_{m}(\alpha )$, shows how the distribution of $y$ differs from the one for
free particles.

The moment generating function contains all the information about the
original distribution. In particular, the cumulants of $y$ are given by the
derivatives of its logarithm at $\alpha =0$,
\begin{equation}
\left\langle y^{n}\right\rangle _{c}=\frac{\partial ^{n}}{\partial \alpha
^{n}}\left. \log \left\langle e^{\alpha y}\right\rangle _{TASEP}\right\vert
_{\alpha =0}.
\end{equation}%
The value of the first derivative, i.e.,%
\begin{equation}
\left\langle y\right\rangle _{c}=\left\langle y\right\rangle =f_{m}^{\prime
}\left( 0\right) ,
\end{equation}%
shows how the difference between the mean velocity $v_{m}$ of the center of
mass of the particles and that of free non-interacting particles, which is $%
\widetilde{p}$, decays with time $t$,
\begin{equation}
v_{m}=\frac{\left\langle Y_{t}\right\rangle }{mt}\simeq \widetilde{p}+\frac{%
\left\langle y\right\rangle _{c}}{\sqrt{t}}\frac{\sqrt{\widetilde{p}\left( 1-%
\widetilde{p}\right) }}{m}\left\langle y\right\rangle _{c}.
\end{equation}%
A nonzero value of $f_{m}^{\prime }\left( 0\right) $ implies that this
difference is of order of $t^{-1/2}$. As the TASEP interaction slows down
the particle motion, one expects it to be negative, i.e.
\begin{equation}
f_{m}^{\prime }\left( 0\right) <0.
\end{equation}%
The second cumulant
\begin{equation}
\left\langle y^{2}\right\rangle _{c}=\left\langle y^{2}\right\rangle
-\left\langle y\right\rangle ^{2}=m+f_{m}^{\prime \prime }\left( 0\right)
-\left( f_{m}^{\prime }(0)\right) ^{2}
\end{equation}%
is related to the diffusion constant $\Delta _{m}$ of the center of mass.
\begin{equation}
\Delta _{m}\equiv \frac{1}{m^{2}}\lim_{t\rightarrow \infty }\frac{%
\left\langle Y_{t}^{2}\right\rangle -\left\langle Y_{t}\right\rangle ^{2}}{t}%
=\frac{\widetilde{p}\left( 1-\widetilde{p}\right) }{m^{2}}\left\langle
y^{2}\right\rangle _{c}
\end{equation}%
The next cumulants, e.g.
\begin{eqnarray}
\left\langle y^{3}\right\rangle _{c} &\equiv &\left\langle
y^{3}\right\rangle -3\left\langle y^{2}\right\rangle \left\langle
y\right\rangle +2\left\langle y\right\rangle ^{3}=\left. \left( \log
f_{m}\left( \alpha \right) \right) ^{\prime \prime \prime }\right\vert
_{\alpha =0}, \\
\left\langle y^{4}\right\rangle _{c} &\equiv &\left\langle
y^{4}\right\rangle -4\left\langle y^{3}\right\rangle \left\langle
y\right\rangle -3\left\langle y^{2}\right\rangle ^{2} \\
&&+12\left\langle y\right\rangle ^{2}\left\langle y^{2}\right\rangle
-6\left\langle y\right\rangle ^{4}=\left. \left( \log f_{m}\left( \alpha
\right) \right) ^{\left( 4\right) }\right\vert _{\alpha =0},  \nonumber
\end{eqnarray}%
quantify the discrepancy of the distribution from a Gaussian form, being
identically zero for the latter.

The asymptotical behavior of the generating function at large absolute
values of $\Re \alpha $ can be readily obtained from the ones of $%
f_{m}(\alpha )$, (\ref{alpha->+infty}),(\ref{alpha->-infty}).
\begin{equation}
\left\langle e^{\alpha y}\right\rangle _{TASEP}\simeq \left\{
\begin{array}{lc}
\alpha ^{-\frac{1}{2}m(m-1)}e^{\frac{1}{2}m\alpha ^{2}}\frac{%
2^{m}\prod\nolimits_{l=1}^{m}\Gamma \left( l/2+1\right) }{\pi ^{m/2}m!}, &
\Re \alpha \rightarrow \infty \\
e^{\frac{1}{2}\alpha ^{2}m^{2}}\frac{m^{m-1}}{\left( m-1\right) !}, & \Re
\alpha \rightarrow -\infty%
\end{array}%
\right.  \label{genfun assympt}
\end{equation}%
The PDF of the random variable $y$ can be obtained as an inverse Laplace
transform of its moment generating function (\ref{generating_transition})
\begin{equation}
P_{m}(y)=\int\limits_{\beta -\mathrm{i\infty }}^{\beta +\mathrm{i\infty }%
}e^{m\alpha ^{2}/2-\alpha y}f_{m}\left( \alpha \right) \frac{d\alpha }{2\pi
\mathrm{i}}  \label{inverselaplace}
\end{equation}%
As the function $f_{m}\left( \alpha \right) $ is bounded and analytic in any
vertical strip of finite width, the parameter $\beta $ can be chosen
arbitrarily. The asymptotic results (\ref{genfun assympt}) for the
generating function can be used in the integral (\ref{inverselaplace}) to
evaluate the asymptotics for PDF $P_{m}(y).$ Choosing $\beta =y/m$ for $%
y\rightarrow \infty $ and $\beta =y/m^{2}$ for $y\rightarrow -\infty $ we
obtain
\begin{equation}
P_{m}(y)\simeq \left\{
\begin{array}{lc}
\left( \frac{m}{y}\right) ^{\frac{m(m-1)}{2}}e^{-\frac{y^{2}}{2m}}\frac{%
\prod\nolimits_{l=1}^{m}\Gamma \left( l/2+1\right) }{m!\sqrt{2\pi m}}\frac{%
2^{m}}{\pi ^{m/2}}, & y\rightarrow \infty \\
e^{-\frac{y^{2}}{2m^{2}}}\frac{m^{m-2}}{\sqrt{2\pi }\left( m-1\right) !}, &
y\rightarrow -\infty%
\end{array}%
\right. .  \label{tails}
\end{equation}%
Thus, the form of the distribution $P_{m}(y)$ is far from being symmetric,
having tails of two Gaussian-like functions with different dispersions, $%
m^{2}$ and $m$, on the left and right respectively, the latter also
multiplied by "Fisher's factor" $y^{-m(m-1)/2}$.

Let us compare these result with the data obtained from Monte Carlo
simulations. We modelled the TASEP for $m=2,3,4,5$ particles, which have
evolved for $t=10^{6}$ time-steps, the statistics having been collected from
$10^{6}$ samples. We would like to compare the data obtained for the
generating function $\left\langle e^{\alpha y}\right\rangle $ and the PDF $%
P_{m}\left( y\right) $ with our predictions. An explicit evaluation of these
functions requires detailed analysis of the function $f_{m}\left( \alpha
\right) $, which is given by the multiple integral (\ref{f_m(alpha)}). For
arbitrary $m$ this needs a significant calculational effort, which is beyond
the goals of the present article. Fortunately, for $m=2$ the function $%
f_{2}\left( \alpha \right) $ is simple enough, being given by (\ref%
{f_2(alpha)}), and we can use it for plotting the generating function and
the distribution. In Fig.~\ref{genfun fig} we show a plot of the logarithm
of the $m=2$ moment generating function, whose analytic expression is
\begin{equation}
\left\langle e^{\alpha y}\right\rangle _{TASEP}=e^{2\alpha ^{2}}\left( 1-%
\mathrm{Erf}\left( \alpha \right) \right) .  \label{genfun_m=2}
\end{equation}%
It has a skew convex form growing more rapidly to the left than to the
right, with a minimum at $\alpha =0.432752$. One can see good agreement with
the numerical data in the central part of the graph. There is some
discrepancy at the tails, which can be attributed to the finite-time
corrections, i.e. the lack of statistics of large events at the finite
period of measurement, which becomes significant when the absolute value of $%
\alpha $ is large.

\begin{figure}[h]
\unitlength=1mm \makebox(110,50)[cc] {\psfig{file=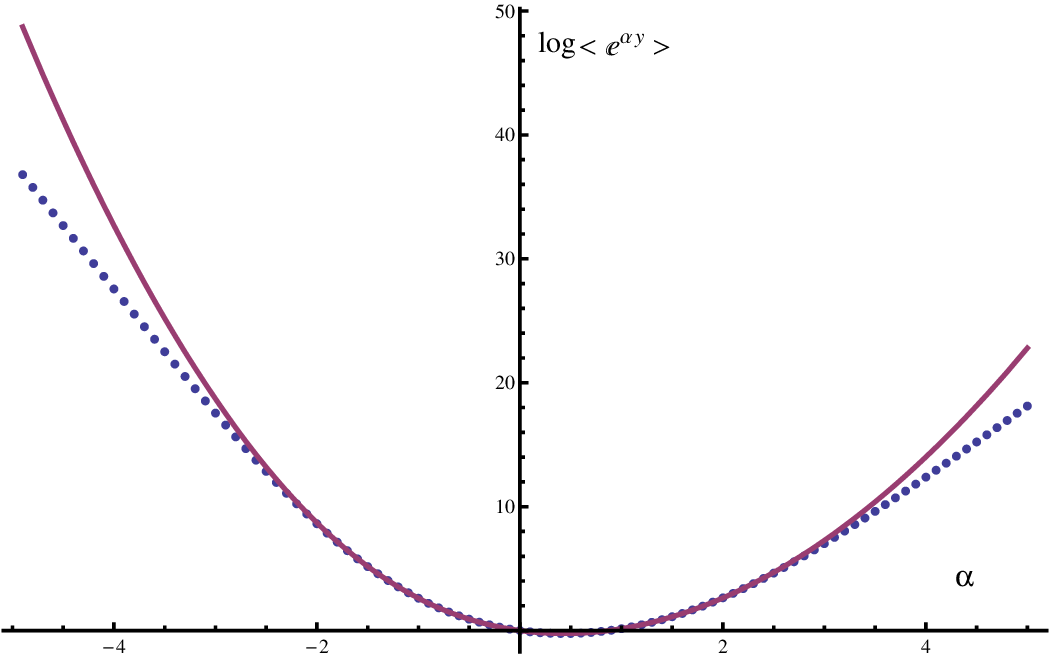,width=80mm}}
\caption{Plot of the logarithm of the moment generating function for $m=2$.
Solid line is the plot of the formula (\protect\ref{genfun_m=2}). Dotted
line is the result of Monte Carlo simulations. }
\label{genfun fig}
\end{figure}
The function (\ref{genfun_m=2}) allows a calculation of any derivatives,
and, hence, of any cumulants of the random variable $y$. In Table~\ref%
{cumulants_table} we show the first four cumulants for $m=2$, the case of
two particles. Their values are in good agreement with the results from the
Monte Carlo simulations.
\begin{table}[h]
{\hspace{1.5cm}}
\begin{tabular}{|l|l|l|l|}
\hline
& {\small Analytic} & {\small Numerical Analytic } & {\small Monte Carlo} \\
\hline
$\left\langle y\right\rangle _{c}$ & $-2\pi ^{-1/2}$ & -1.12838 & -1.12545
\\ \hline
$\left\langle y^{2}\right\rangle _{c}$ & $4-4\pi ^{-1}$ & 2.72676 & 2.72518
\\ \hline
$\left\langle y^{3}\right\rangle _{c}$ & $4(\pi -4)\pi ^{-3/2}$ & -0.616636
& -0.617642 \\ \hline
$\left\langle y^{4}\right\rangle _{c}$ & $32(\pi -3)\pi ^{-2}$ & 0.459083 &
0.498263 \\ \hline
\end{tabular}%
\label{cumulants_table}
\caption{Cumulants of the random variable $y$.}
\end{table}
In Fig.~\ref{P_2 fig} we show the result of numerical evaluation of the
integral (\ref{inverselaplace}) for $m=2$.
\begin{figure}[h]
\unitlength=1mm \makebox(120,50)[cc]
{\psfig{file=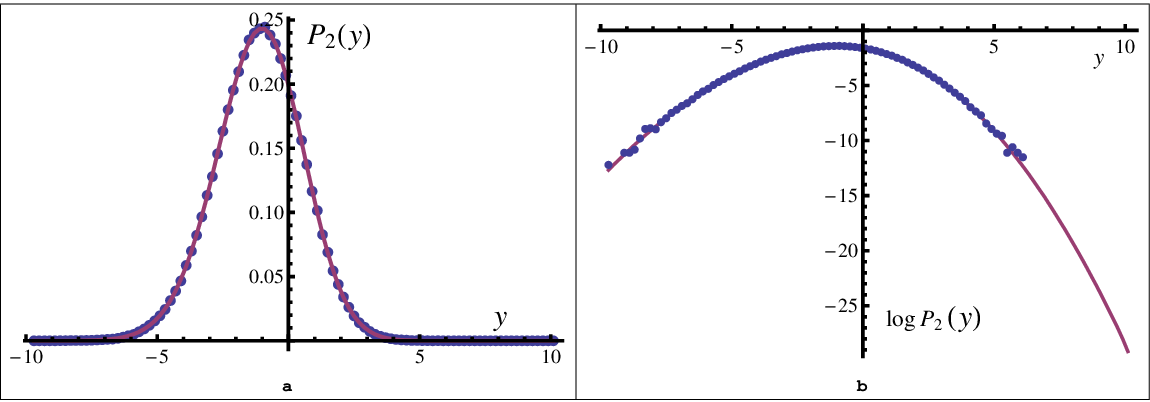,width=120mm}}
\caption{Probability density function $P_{m}(y)$ for $m=2$ particles (a) and
its logarithm (b). The solid line shows the theoretical predictions for the
distributions.}
\label{P_2 fig}
\end{figure}
There is a very good agreement with the simulation results. At first glance
the form of the distribution shown on Fig.~\ref{P_2 fig}a appears
Gaussian-like. A more accurate impression of the form of the distribution is
given by the logarithmic plot of Fig.~\ref{P_2 fig}b which shows that the
distribution is actually skew, decreasing more rapidly as $y$ grows than as
it decreases.

Simulation results obtained for more than two particles can be tested
against the asymptotical formulas (\ref{tails}) for the tails of the
distribution $P_{m}\left( y\right) $. In Fig.~\ref{many particles} we plot
the distributions measured for $m=2,3,4,5$ particles, (Fig.~\ref{many
particles}a), and its logarithm, (Fig.~\ref{many particles}b), the latter
being compared with the graphs of (\ref{tails}).

\begin{figure}[h]
\unitlength=1mm
\makebox(120,50)[cc]
{\psfig{file=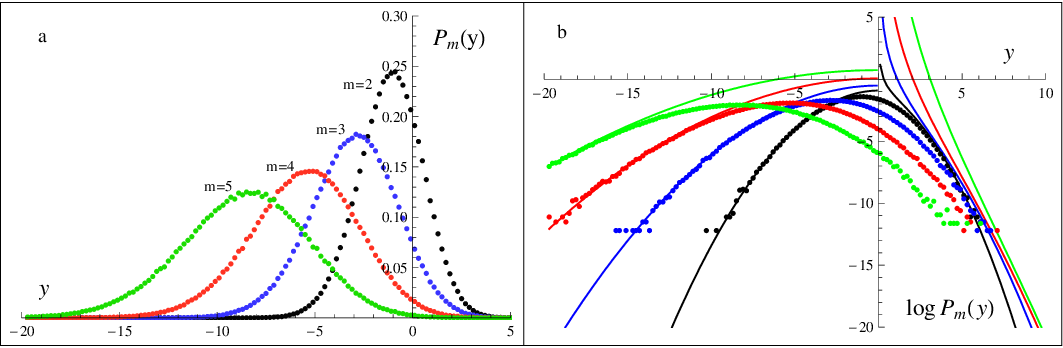,width=120mm}}
\caption{Probability density function $P_{m}(y)$ for $m=2,3,4,5$ particles
(a) and its logarithm (b). The solid lines show the theoretical asymptotics
of the tails of the distributions.}
\label{many particles}
\end{figure}

One can see that for all the four graphs the left tails are perfectly fitted
already for rather small values of $y$. A good fit of the right tail takes
place only for $m=2$. For $m=3$ only a few data points approach the
asymptotical line, i.e. in this case the right asymptotical regime is
actually at the borderline of the statistics. In the cases $m=3,4$ the
statistics available is clearly not good enough to reach the asymptotical
regime. Significantly larger evolution time and statistics would have to be
considered.

\section{ The master equation}

From now on we consider only the SVW model dependent on the parameters $p$
and $\kappa .$ Our first step is to calculate the probability $P_{t}\left(
X,X^{0}\right) $ of transition from the configuration $X^{0}$ to $X$ for
arbitrary time $t$:
\begin{equation}
P_{t}\left( X,X^{0}\right) =\sum\limits_{X^{0}\leq X^{1}\cdots \leq
X^{t}\equiv X}P(X^{0},\ldots ,X^{t}).
\end{equation}%
The method of finding the transition probability was first developed by Sch%
\"{u}tz for the continuous time TASEP \cite{Schuetz TASEP}, who used the
Bethe Ansatz first applied to the ASEP by Gwa and Spohn \cite{Gwa Spohn}.
Here we follow a similar procedure. The transition probability obeys the
master equation
\begin{equation}
P_{t}\left( X,X^{0}\right) =\sum\limits_{X^{\prime }}T\left( X,X^{\prime
}\right) P_{t-1}(X^{\prime },X^{0});
\end{equation}%
the transition weights $T\left( X,X^{\prime }\right) $ being defined as
above, (\ref{weights})-(\ref{weights_12}). The problem of finding the
eigenvectors and eigenvalues of the matrix $T\left( X,X^{\prime }\right) $
can be solved by the Bethe Ansatz technique. As this technique is rather
standard and has been reviewed in many monographs, we simply state the
results here. For details of similar derivations, the reader can consult for
example with the review \cite{Derrida}. As a result we obtain the solution
of the left and right eigenvalue problems for the Markov matrix $T\left(
X,X^{\prime }\right) $:
\begin{eqnarray}
\Lambda \left( Z\right) \Psi _{Z}\left( X\right) &=&\sum\limits_{X^{\prime
}}T\left( X,X^{\prime }\right) \Psi _{Z}\left( X^{\prime }\right) , \\
\Lambda \left( Z\right) \overline{\Psi }_{Z}\left( X\right)
&=&\sum\limits_{X^{\prime }}T\left( X^{\prime },X\right) \overline{\Psi }%
_{Z}\left( X^{\prime }\right)
\end{eqnarray}%
parametrized by an $m$-tuple of complex parameters $Z=\left\{ z_{1},\cdots
,z_{m}\right\} .$ The corresponding eigenvalue is expressed in terms of
these parameters,
\begin{equation}
\Lambda \left( Z\right) =\prod\limits_{i=1}^{m}\left( 1-p+p/z_{i}\right) ,
\label{Lambda(Z)}
\end{equation}%
and the eigenvectors are given by the following determinants
\begin{eqnarray}
\Psi _{Z}\left( X\right) &=&\det \left( z_{i}^{x_{j}}(1-\kappa
z_{i})^{i-j}\right) _{1\leq i,j\leq m}, \\
\overline{\Psi }_{Z}\left( X\right) &=&\det \left( z_{i}^{-x_{j}}(1-\kappa
z_{i})^{j-i}\right) _{1\leq i,j\leq m}.
\end{eqnarray}%
It is not difficult to check that these two eigenfunctions can be used to
construct the resolution of the identity operator
\begin{equation}
\frac{1}{m!}\oint \Psi _{Z}\left( X\right) \overline{\Psi }_{Z}\left(
X^{\prime }\right) \prod\limits_{i=1}^{m}\frac{dz_{i}}{2\pi \mathrm{i}z_{i}}%
=\delta _{X,X^{\prime }},  \label{resolution}
\end{equation}%
where the integration over each $z_{i},i=1,\ldots ,m$, is along a contour of
integration which has to satisfy the requirement that the pole of the wave
function at $z=1/\kappa $ has to lie in the exterior. Then the solution of
the initial value problem for the master equation is given by%
\begin{equation}
P_{t}\left( X,X^{0}\right) =\frac{1}{m!}\oint \Lambda ^{t}\left( Z\right)
\Psi _{Z}\left( X\right) \overline{\Psi }_{Z}\left( X^{0}\right)
\prod\limits_{i=1}^{m}\frac{dz_{i}}{2\pi iz_{i}}.
\label{transition prob integral}
\end{equation}%
Finally we end up with the following integral expression for the transition
probability 
\begin{eqnarray}
P_{t}\left( X,X^{0}\right) &=&\oint \Lambda ^{t}\left( Z\right)
\prod\limits_{i=1}^{m}\left[ \frac{z_{i}^{x_{i}-x_{m}^{0}}}{(1-\kappa
z_{i})^{i-1}}\right]  \label{int} \\
&&\times \det \left( z_{i}^{x_{m}^{0}-x_{j}^{0}}(1-\kappa
z_{i})^{j-1}\right) _{1\leq i,j\leq m}\prod\limits_{i=1}^{p}\frac{dz_{i}}{%
2\pi \mathrm{i}z_{i}}.  \nonumber
\end{eqnarray}%
The integration can be easily performed by counting the residues. The result
is a determinant of an $m\times m$ matrix of the form similar to the one
obtained for the discrete time TASEP with backward update ( \cite{Brankov
Priezzhev Shelest} , \cite{Rakos Schuetz}). Note that in the case of vicious
walkers, $\kappa =0$, the eigenfunctions are of free fermion type%
\begin{eqnarray}
\Psi _{Z}\left( X\right) &=&\det \left( z_{i}^{x_{j}}\right) _{1\leq i,j\leq
m}, \\
\overline{\Psi }_{Z}\left( X\right) &=&\det \left( z_{i}^{-x_{j}}\right)
_{1\leq i,j\leq m}.
\end{eqnarray}%
and the integration yields the famous Lindstr\"{o}m-Gessel-Viennot theorem
\cite{Lindstrom}, \cite{Gessel Viennot}.
\begin{equation}
P_{t}\left( X,X^{0}\right) =\det \left[ F_{0}\left( x_{i}-x_{j}^{0},t\right) %
\right] _{1\leq i,j\leq m},  \label{Gessel Viennot}
\end{equation}%
where%
\begin{equation}
F_{0}\left( x,t\right) =p^{x}\left( 1-p\right) ^{t-x}\left(
\begin{array}{c}
t \\
x%
\end{array}%
\right) ,
\end{equation}%
These formulas serve as a starting point for the asymptotical analysis of
the survival probability.

\section{Asymptotic form of the survival probability}

To obtain the survival probability $\mathcal{P}_{t}\left( X^{0}\right) $ for
SVW we have to sum the transition probability $P_{t}\left( X,X^{0}\right) $
over the set of all final configurations $X$:
\begin{equation}
\mathcal{P}_{t}\left( X^{0}\right) =\sum\limits_{\left\{ X\right\}
}P_{t}\left( X,X^{0}\right) .  \label{survival probability def}
\end{equation}%
We solve this problem in the limit $t\rightarrow \infty $. For pedagogical
reasons we first outline the derivation for the VW model, which simply
reproduce Rubey's results, \cite{Rubey}. The procedure we use amounts to an
asymptotical analysis of the expression for $P_{t}\left( X,X^{0}\right) $ by
means of the saddle point approximation for the integral (\ref{transition
prob integral}), which reduces the sum over final configurations to known
integrals. The main ingredients of the derivation for the VW model are then
applied similarly to the SVW model but with some modifications.

\subsection{Vicious walkers}

In the case of VW\ ($\kappa =0$), the integral (\ref{int}) takes the form
\begin{equation}
\oint \Lambda ^{t}\left( Z\right)
\prod\limits_{k=1}^{m}z_{k}^{x_{k}-x_{m}^{0}}\det
(z_{i}^{x_{m}^{0}-x_{j}^{0}})_{1\leq i,j\leq m}\prod\limits_{l=1}^{m}\frac{%
dz_{l}}{2\pi \mathrm{i}z_{l}}.  \label{viciouswalkerint}
\end{equation}%
Here, the determinant under the integral can be expressed in terms of the
Vandermonde determinant
\begin{equation}
\Delta \left( Z\right) \equiv \det (z_{i}^{m-j})_{1\leq i,j\leq
m}=\prod\limits_{1\leq i<j\leq m}\left( z_{i}-z_{j}\right) ,
\end{equation}%
and the Schur function \cite{Macdonald}
\begin{equation}
s_{\chi }(z_{1},\ldots ,z_{m})\equiv \det \left( z_{i}^{\chi
_{j}+m-j}\right) /\Delta \left( Z\right)
\end{equation}%
parametrized by the partition $\chi =\left( \chi _{1}\geq \chi _{2}\geq
\cdots \geq \chi _{m}\geq 0\right) $ defined by
\begin{equation}
\chi =(x_{m}^{0}-x_{1}^{0}-m+1,x_{m}^{0}-x_{2}^{0}-m+2,\cdots ),  \label{khi}
\end{equation}%
as follows:
\begin{equation}
\det (z_{i}^{x_{m}^{0}-x_{i}^{0}})_{1\leq i,j\leq m}=\Delta \left( Z\right)
s_{\chi }\left( Z\right) .
\end{equation}

Thus (\ref{viciouswalkerint}) can be rewritten in the following form
\begin{equation}
P_{T}\left( X,X^{0}\right) =\oint \Delta \left( Z\right) s_{\chi }\left(
Z\right) \prod\limits_{i=1}^{m}e^{th_{i}\left( z_{i}\right) }\frac{dz_{i}}{%
2\pi \mathrm{i}z_{i}},
\end{equation}%
where
\begin{equation}
h_{i}(z)=\log \left( 1-p+p/z\right) +v_{i}\log z  \label{f_i(z)}
\end{equation}%
and
\begin{equation}
v_{i}=\frac{x_{i}-x_{m}^{0}}{t}.
\end{equation}%
Now we are ready to estimate the integral asymptotically as $t\rightarrow
\infty $. We assume that the differences $\left( x_{i}^{0}-x_{j}^{0}\right) $
are kept bounded for any $i$ and $j$. The saddle point of the function under
the integral is defined by the equation
\begin{equation}
h_{i}^{\prime }(z_{i}^{\ast })=0,
\end{equation}%
which yields
\begin{equation}
z_{i}^{\ast }=\frac{\left( 1-v_{i}\right) p}{\left( 1-p\right) v_{i}}.
\label{z^* saddle point}
\end{equation}%
In the vicinity of the saddle point $h_{i}(z)$ has an expansion%
\begin{eqnarray}
h_{i}(z_{i}^{\ast }+\xi ) &=&\log \left[ \left( \frac{1-p}{1-v_{i}}\right)
^{1-v_{i}}\left( \frac{p}{v_{i}}\right) ^{v_{i}}\right] \\
&&+\frac{1}{2}\left( \frac{1-p}{p}\right) ^{2}\frac{v_{i}^{3}}{1-v_{i}}\xi
^{2}+O(\xi ^{3}).  \nonumber
\end{eqnarray}%
The integration contours can be deformed to a circle centered at 0, crossing
the real axis at $z_{i}^{\ast }$. Writing points on the circle as $%
z_{i}=z_{i}^{\ast }e^{\mathrm{i}\phi _{i}}$, we have
\begin{equation}
\Re (h_{i}(z_{i}))=h(z_{i}^{\ast })+\frac{1}{2}\log \left[ (1-v_{i}(1-\cos
(\phi _{i}))^{2}+v_{i}^{2}\sin ^{2}(\phi _{i})\right] .
\end{equation}%
It follows that there is a single maximum at $\phi _{i}=0$. Moreover, since
all derivatives are bounded provided $v_{i}<0$, the saddle-point
approximation holds uniformly in $x_{i}$. (Note that (\ref{viciouswalkerint}%
) is zero if $v_{i}>1$, and if $v_{m}=1$ then the probability $p^{t}$ can be
extracted as a factor, the remaining integral over $z_{1},\dots ,z_{m-1}$
being of the same form.) It is easy to see that the contribution to the sum
over $X$ from points with $v_{i}>p+\epsilon $ for any fixed $\epsilon >0$ is
negligible in the limit $t\rightarrow \infty $. (Below, we shall see that
the effective range of the summation is in fact even smaller.) The
saddle-point approximation \cite{Fedoryuk} now yields
\begin{eqnarray}
&&P_{t}\left( X,X^{0}\right) =\left( \frac{p}{1-p}\right) ^{\frac{m(m-1)}{2}%
}\prod\limits_{i=1}^{m}\frac{v_{i}^{1-m}e^{th_{i}\left( \frac{\left(
1-v_{i}\right) p}{\left( 1-p\right) v_{i}}\right) }}{\sqrt{2\pi
tv_{i}(1-v_{i})}}\times  \label{P_T(X,X^0)_1} \\
&&\prod\limits_{1\leq i<j\leq m}\left( v_{j}-v_{i}\right) s_{\chi }\left(
\frac{\left( 1-v_{1}\right) p}{\left( 1-p\right) v_{1}},\ldots ,\frac{\left(
1-v_{m}\right) p}{\left( 1-p\right) v_{m}}\right) \left( 1+O\left( \frac{1}{t%
}\right) \right) .  \nonumber
\end{eqnarray}%
The next step is to perform the summation (\ref{survival probability def})
over the range of the final configurations $X\in \left\{ x_{1}^{0}\leq
x_{1}<\cdots <x_{m}<\infty \right\} $. For this we need to demonstrate that (%
\ref{P_T(X,X^0)_1}) holds uniformly in $X$. To this end we first show that
the main contribution to the sum comes from the domain
\begin{equation}
pt-\sqrt{t}\log t\leq x_{1}<\cdots <x_{m}\leq \ pt+\sqrt{t}\log t.
\label{summation range}
\end{equation}%
Indeed, $h_{i}\left( \frac{\left( 1-v_{i}\right) p}{\left( 1-p\right) v_{i}}%
\right) $ is a concave function of $v_{i}$ in the domain $v_{i}\in (0,1)$
with a single maximum $v_{i}=p$. It follows then for $\left\vert
x_{i}-pt\right\vert >\sqrt{t}\log t$
\begin{equation}
e^{th_{i}\left( \frac{\left( 1-v_{i}\right) p}{\left( 1-p\right) v_{i}}%
\right) }<e^{th_{i}\left( \frac{1-\sqrt{t}\log t/(1-p)}{1+\sqrt{t}\log t/p}%
\right) }=e^{-\frac{\left( \log t\right) ^{2}}{2p\left( 1-p\right) }}\left[
1+O\left( \frac{\log t}{\sqrt{t}}\right) \right]
\end{equation}%
All the other factors in (\ref{P_T(X,X^0)_1}) are at most of polynomial
order in $t$, while the total number of nonzero terms in the sum of interest
(\ref{survival probability def}) is $O\left( t^{m}\right) $. Therefore, the
contribution from the complement of (\ref{summation range}) being of order
of $O\left( t^{s}e^{-\frac{\left( \log t\right) ^{2}}{2p\left( 1-p\right) }%
}\right) $ for some constant $s$ is asymptotically negligible compared to
the contribution from (\ref{summation range}).


In the latter one can approximate the function $h_{i}\left( \frac{\left(
1-v_{i}\right) p}{\left( 1-p\right) v_{i}}\right) $ by the second term of
its Taylor expansion at $v_{i}=p$, which yields
\begin{eqnarray}
P_{t}\left( X,X^{0}\right) &=&\frac{1}{(2\pi )^{\frac{m}{2}}\left(
tp(1-p)\right) ^{\frac{m^{2}}{2}}}s_{\chi }\left( 1,\ldots ,1\right) \\
&&\hspace{-2cm}\prod\limits_{i=1}^{m}\exp \left( -\frac{\left(
x_{i}-x_{m}^{0}-pt\right) ^{2}}{2tp\left( 1-p\right) }\right)
\prod\limits_{1\leq i<j\leq m}\left( x_{j}-x_{i}\right) \left[ 1+O\left(
\frac{\left( \log t\right) ^{3}}{\sqrt{t}}\right) \right] .  \nonumber
\label{P_t(X,X)_2}
\end{eqnarray}%
We now have to evaluate following sum in the limit $t\rightarrow \infty $,
\begin{equation}
\sum_{-\sqrt{t}\log t\leq x_{1}<\cdots <x_{m}\leq \ \sqrt{t}\log
t}\prod\limits_{i=1}^{m}e^{-\frac{x_{i}^{2}}{2tp\left( 1-p\right) }%
}\prod\limits_{1\leq i<j\leq m}\left( x_{j}-x_{i}\right) .  \label{sum_VW}
\end{equation}%
This can be done by means of the following lemma:

\begin{lemma}
\label{sum-to-integral} Let $h:\mathbb{R}^{m}\rightarrow \mathbb{R}$ be a
twice differentiable function of at most polynomial growth. Then as $\delta
\rightarrow 0$,
\begin{eqnarray}
&&\delta ^{m}\sum\limits_{{y_{1}<\dots <y_{m}};\,{y_{i}\in \delta \mathbb{Z}}%
}h(y_{1},\dots ,y_{n})\prod_{i=1}^{m}e^{-\frac{1}{2}y_{i}^{2}} \\
&=&\int_{-\infty }^{\infty }dy_{1}\int_{y_{1}}^{\infty }dy_{2}\dots
\int_{y_{m-1}}^{\infty }dy_{m}h(y_{1},\dots ,y_{m})e^{-\frac{1}{2}%
\sum_{i=1}^{m}y_{i}^{2}}+O(\delta ).  \nonumber
\end{eqnarray}
\end{lemma}

\begin{proof}
We subdivide the domain $-\infty <x_{1}<\dots <x_{m}<\infty $ into
hypercubes of the form
\begin{equation}
B_{\delta }(y_{1},\dots ,y_{m})=\{(x_{1},\dots ,x_{m}):\,\max
|x_{i}-y_{i}|\leq \frac{\delta }{2}\},
\end{equation}%
where $y_{i}\in \delta \mathbb{Z}$ and $y_{1}<y_{2}<\dots <y_{m}$. The
remaining region is small and its contribution will be estimated shortly. We
then write, for $(x_{1},\dots ,x_{m})\in B_{\delta }(y_{1},\dots ,y_{m})$,
\begin{eqnarray}
&&\left\vert h(x_{1},\dots ,x_{m})e^{-\frac{1}{2}(x_{1}^{2}+\dots
+x_{m}^{2})}-h(y_{1},\dots ,y_{m})e^{-\frac{2}{2}(y_{1}^{2}+\dots
+y_{m}^{2})}\right\vert \\
&\leq &\sup_{(u_{1},\dots ,u_{m})\in B_{\delta }(y_{1},\dots
,y_{m})}\max_{i=1}^{m}\left\vert \frac{\partial }{\partial u_{i}}%
h(u_{1},\dots ,u_{m})e^{-\frac{1}{2}(u_{1}^{2}+\dots +u_{m}^{2})}\right\vert
\,|x_{i}-y_{i}|.  \nonumber
\end{eqnarray}%
Now,
\begin{eqnarray}
&&\frac{\partial }{\partial u_{i}}h(u_{1},\dots ,u_{m})e^{-\frac{1}{2}%
(u_{1}^{2}+\dots +u_{m}^{2})} \\
&=&\left( \frac{\partial }{\partial u_{i}}h(u_{1},\dots
,u_{m})-u_{i}h(u_{1},\dots ,u_{m})\right) e^{-\frac{1}{2}(u_{1}^{2}+\dots
+u_{m}^{2})}  \nonumber
\end{eqnarray}%
which is easily seen to be bounded by $Ce^{-(y_{1}^{2}+\dots +y_{m}^{2})/2}$
for some constant $C>0$. It follows from the convergence of the sum $%
\sum_{y\in \delta \mathbb{Z}}\delta e^{-y^{2}/2}$ uniformly in $\delta $
that the difference between the integral over the region
\begin{equation}
\bigcup_{y_{1}<\dots <y_{m};\,y_{i}\in \delta \mathbb{Z}}B_{\delta
}(y_{1},\dots ,y_{m})
\end{equation}%
and the sum is of order $\delta $. There remains the integral over the
complementary region, but this is obviously of order $\delta $ as the
integral converges and the region has width $\delta $.
\end{proof}


Then, after going to rescaled variables $y_{i}=x_{i}/\sqrt{tp\left(
1-p\right) }$ and writing $\delta =1/\sqrt{tp(1-p)}$ the sum (\ref{sum_VW})
reduces to the integral
\begin{equation}
\left( tp\left( 1-p\right) \right) ^{\frac{m(m+1)}{4}}\int\limits_{-\infty
}^{\infty }dx_{1}\cdots \int\limits_{x_{m-2}}^{\infty
}dx_{m-1}\int\limits_{x_{m-1}}^{\infty }dx_{m}\prod\limits_{i=1}^{m}e^{-%
\frac{1}{2}x_{i}^{2}}\prod\limits_{1\leq i<j\leq m}\left\vert
x_{j}-x_{i}\right\vert
\end{equation}%
(the range of summation is extended to ($-\infty \leq x_{1}<\cdots
<x_{m}\leq \infty $) by the same argument as above. Note that the absolute
value signs $\left\vert x_{j}-x_{i}\right\vert $, though redundant in this
range, are nevertheless useful as they make the expression symmetric with
respect to permutations of the variables $x_{1},\ldots ,x_{m}$. One, then,
can use this fact to extend the integration to the whole $\mathbb{R}^{m}$,
which yields an additional factor of $m!$, which has to be compensated in
the end. As a result we obtain
\begin{equation}
\mathcal{P}_{t}\left( X^{0}\right) =\frac{1}{\left[ p(1-p)t\right] ^{m\left(
m-1\right) /4}}\frac{I_{m,1/2}}{\left( 2\pi \right) ^{m/2}m!}s_{\chi }\left(
1,\ldots ,1\right) \left[ 1+O\left( \frac{\left( \log t\right) ^{3}}{\sqrt{t}%
}\right) \right]
\end{equation}%
where
\begin{eqnarray}
I_{m,k} &\equiv &\int\limits_{-\infty }^{\infty }dy_{p}\cdots
\int\limits_{-\infty }^{\infty }dy_{2}\int\limits_{-\infty }^{\infty
}dy_{1}\exp \left( -\frac{1}{2}\sum\nolimits_{i=1}^{m}y_{i}^{2}\right)
\nonumber \\
&&\hspace{5cm}\times \prod\limits_{1\leq i<j\leq m}\left\vert
y_{j}-y_{i}\right\vert ^{2k}  \nonumber \\
&=&\left( 2\pi \right) ^{m/2}\prod\limits_{l=1}^{m}\frac{\Gamma \left(
lk+1\right) }{\Gamma (k+1)}  \label{Mehta}
\end{eqnarray}%
is the Mehta integral \cite{Mehta Dyson}, which first appeared in the
context of Gaussian random matrix ensembles. Finally, one can use the
following formula for the Schur function \cite{Macdonald}%
\begin{equation}
s_{\chi }\left( 1,\ldots ,1\right) =\prod\limits_{1\leq i<j\leq m}\frac{\chi
_{i}-i-\chi _{j}+j}{j-i},
\end{equation}%
resulting in the following expression for the survival probability:%
\begin{eqnarray}
\mathcal{P}_{t}\left( X^{0}\right) &=&\frac{1}{\left[ p(1-p)t\right]
^{m\left( m-1\right) /4}}\frac{2^{m}}{\pi ^{m/2}}\prod\limits_{l=1}^{m}\frac{%
\Gamma \left( l/2+1\right) }{l!}  \label{P^vw} \\
&&\times \prod\limits_{1\leq i<j\leq m}\left( x_{j}^{0}-x_{i}^{0}\right)
\left[ 1+O\left( \frac{\left( \log t\right) ^{3}}{\sqrt{t}}\right) \right]
\nonumber
\end{eqnarray}%
After reexpressing the gamma functions in terms of factorials we obtain the
form given in (\ref{A(X^0;0)}). \medskip

\subsection{Semi-vicious walkers}

\subsubsection{The case of generic $\protect\kappa \neq 1$}

To study the asymptotic behaviour of the survival probability for the case
of general $\kappa $, one can start with the following integral
representation for the transition probability
\begin{eqnarray}
P_{t}\left( X,X^{0}\right) &=&\prod\limits_{i=1}^{m}\oint_{C_{0}}\frac{dz_{i}%
}{2\pi \mathrm{i}z_{i}}\frac{d\xi _{i}}{2\pi \mathrm{i}\xi _{i}}\left( \frac{%
1-\kappa \xi _{i}}{1-\kappa z_{i}}\right) ^{i-1}\xi
_{i}^{x_{m}^{0}-x_{i}^{0}+1}z_{i}^{x_{i}-x_{m}^{0}}  \nonumber \\
&&\times \Lambda ^{t}\left( Z\right) \prod\limits_{1\leq i,j\leq m}\frac{1}{%
\xi _{i}-z_{j}}\prod\limits_{1\leq i<j\leq m}(z_{i}-z_{j})(\xi _{j}-\xi
_{i}),  \label{doubleintegral}
\end{eqnarray}%
where the integration in each variable is along a small circle around zero, $%
\left\vert z_{i}\right\vert <\left\vert \xi _{j}\right\vert $ for any $%
i,j=1,\ldots ,m$. This representation can be reduced to the form (\ref{int})
by direct integration over each $\xi _{j}$ ($j=1,\ldots ,m$). This is done
by summing the contributions to the integral coming from all the poles $\xi
_{j}=z_{i}$, $i=1,\ldots ,m$.

Though the most of analysis of the large $t$ asymptotics of this expression
is similar to the one for VW, one important difference exists. The
expressions under the integrals over $z_{i}$, $i=2,\ldots ,m$, \ $\ $have
singularities at $z_{i}=1/\kappa $, the poles of the form $\left( 1-\kappa
z_{i}\right) ^{1-i}$, which can be located between the origin and the saddle
point. In this case the contour being deformed to the steepest descent one,
crosses this singularity and its contribution must then be extracted from
the saddle point contribution. While for $\left\vert \kappa \right\vert <1$
this does not affect the asymptotics of the sum over $x_{1},\ldots ,x_{m}$,
evaluated subsequently, for $\left\vert \kappa \right\vert >1$ its
contribution turns out to be dominant.

It is, however, technically difficult to calculate the residue at the
multiple pole of the complicated expression. To avoid this calculation and
to evaluate both cases in one go, we expand the term $\left( 1-\kappa
z_{i}\right) ^{1-i}$ into a series in powers of $(\kappa z_{i})$ and then
integrate it term by term in the saddle point approximation. As a result we
obtain $P_{t}\left( X,X^{0}\right) $ in the form of an $(m-1)$ - fold series
\begin{eqnarray}
&&\hspace{-0.5cm}P_{t}\left( X,X^{0}\right)
=\prod\limits_{k=1}^{m}\oint_{C_{0}}\frac{d\xi _{i}}{2\pi \mathrm{i}\xi _{k}}%
\left( 1-\kappa \xi _{k}\right) ^{i-1}\xi _{k}^{x_{m}^{0}-x_{i}^{0}+1}
\label{P^svw_T(X,X0)} \\
&&\prod\limits_{1\leq i<j\leq m}(\xi _{j}-\xi _{i})\sum_{\left\{
n_{2},\ldots ,n_{m}\right\} \in \mathbb{Z}_{\geq 0}^{m-1}}\mathcal{A}\left(
\left\{ \xi _{i},v_{i}\right\} _{i=1}^{m},\left\{ n_{k}\right\}
_{i=2}^{m},\right)  \nonumber
\end{eqnarray}%
where%
\begin{eqnarray}
&&\mathcal{A}\left( \left\{ \xi _{i},v_{i}\right\} _{i=1}^{m},\left\{
n_{k}\right\} _{i=2}^{m}\right) =  \nonumber \\
&&\left( \frac{p}{1-p}\right) ^{\frac{m(m-1)}{2}}\prod\limits_{i=2}^{m}%
\kappa ^{n_{i}}\left(
\begin{array}{c}
i+n_{i}-2 \\
n_{i}%
\end{array}%
\right)  \label{A} \\
&&\times \prod\limits_{i=1}^{m}\frac{v_{i}^{1-m}e^{th_{i}\left( \frac{\left(
1-v_{i}\right) p}{\left( 1-p\right) v_{i}}\right) }}{\sqrt{2\pi
tv_{i}(1-v_{i})}}\prod\limits_{1\leq i<j\leq m}\left( v_{j}-v_{i}\right)
\nonumber \\
&&\times \prod\limits_{1\leq i,j\leq m}\left( \xi _{i}-\frac{\left(
1-v_{j}\right) p}{\left( 1-p\right) v_{j}}\right) ^{-1}\left( 1+O\left(
\frac{1}{t}\right) \right) ,  \nonumber
\end{eqnarray}%
\begin{eqnarray}
v_{1} &=&\frac{x_{1}-x_{m}^{0}}{t}  \label{v_i kappa neq 1} \\
v_{i} &=&\frac{x_{i}-x_{m}^{0}+n_{i}}{t},i=2,\ldots ,m.  \nonumber
\end{eqnarray}%
The next step is to use this approximation to perform the summation of (\ref%
{P^svw_T(X,X0)}) over the domain $\{x_{1}^{0}<x_{1}<x_{2}<\cdots
<x_{m}<\infty \}$. The effective range of this summation depends crucially
on the behaviour of the other sum in $n_{2},\ldots ,n_{m}$. Namely, the
effective summation range is different depending on whether the value of $%
\kappa $ is greater or less than one, when the term $\kappa ^{n_{i}}$ is
decreasing or increasing respectively. We consider these two cases
separately.

\paragraph{The case $\left\vert \protect\kappa \right\vert <1$}

Here $\kappa $ takes arbitrary complex values in the domain $\left\vert
\kappa \right\vert <1$. As in the case of vicious walkers, we argue that the
exponential part $\exp \left[ th_{i}\left( \frac{\left( 1-v_{i}\right) p}{%
\left( 1-p\right) v_{i}}\right) \right] $ makes the whole expression
negligible beyond the range
\begin{equation}
p-t^{-1/2}\log t\leq v_{i}\leq p+t^{-1/2}\log t.  \label{v range}
\end{equation}%
In addition, for $n_{i}>\left( \log t\right) ^{2}/\left\vert \log \left\vert
\kappa \right\vert \right\vert $ we have
\begin{equation}
\kappa ^{n_{i}}\left(
\begin{array}{c}
i+n_{i}-2 \\
n_{i}%
\end{array}%
\right) =O(t^{-\log t}\left( \log t\right) ^{2(i-2)}),
\end{equation}%
so that we can limit the summation over $n_{i}$ to $n_{i}\leq \left( \log
t\right) ^{2}/\left\vert \log \left\vert \kappa \right\vert \right\vert $.
Therefore, $n_{i}$ is negligible compared to $x_{i}$ in the domain (\ref{v
range}) and can be neglected in the definition (\ref{v_i kappa neq 1}) of $%
v_{i}$. In the range (\ref{v range}) we can approximate $\mathcal{A}\left(
\left\{ \xi _{i},x_{i}/t\right\} _{i=1}^{m},\left\{ n_{k}\right\}
_{i=2}^{m},\right) $ by the leading term of its Taylor expansion at $v_{i}=p$%
, $i=1,\ldots ,m$.%
\begin{eqnarray}
&&\mathcal{A}\left( \left\{ \xi _{i},x_{i}/t\right\} _{i=1}^{m},\left\{
n_{k}\right\} _{i=2}^{m}\right) =  \nonumber \\
&&\left( \frac{1}{tp\left( 1-p\right) }\right) ^{\frac{m(m-1)}{2}%
}\prod\limits_{i=2}^{m}\kappa ^{n_{i}}\left(
\begin{array}{c}
i+n_{i}-2 \\
n_{i}%
\end{array}%
\right) \\
&&\times \prod\limits_{i=1}^{m}\frac{e^{-\frac{\left( x_{i}-pt\right) ^{2}}{%
2p(1-p)t}}}{\sqrt{2\pi tp(1-p)}}\prod\limits_{1\leq i<j\leq m}\left(
x_{j}-x_{i}\right)  \nonumber \\
&&\times \prod\limits_{i=1}^{m}\left( \xi _{i}-1\right) ^{-m}\left(
1+O\left( \frac{\left( \log t\right) ^{3}}{\sqrt{t}}\right) \right)
\nonumber
\end{eqnarray}%
One can see that the terms dependent on $\left\{ x_{i}\right\} $ and $%
\left\{ n_{i}\right\} $ decouple and the terms dependent on $n_{2},\ldots
,n_{m}$ can be summed up.
\begin{equation}
\sum_{\left\{ n_{2},\ldots ,n_{m}\right\} \in \mathbb{Z}_{\geq
0}^{m-1}}\prod\limits_{i=2}^{m}\kappa ^{n_{i}}\left(
\begin{array}{c}
i+n_{i}-2 \\
n_{i}%
\end{array}%
\right) =\left( 1-\kappa \right) ^{-\frac{m(m-1)}{2}}
\end{equation}%
The remaining sum over $x_{i}$ $\,$for $i=1,\ldots ,m$ is transformed to an
integral using Lemma \ref{sum-to-integral}:
\begin{eqnarray}
&&\sum\limits_{x_{1}^{0}<x_{1}<x_{2}<\cdots <x_{m}}\sum_{\left\{
n_{2},\ldots ,n_{m}\right\} \in \mathbb{Z}_{\geq 0}^{m-1}}\mathcal{A}\left(
\left\{ \xi _{i},x_{i}/t\right\} _{i=1}^{m},\left\{ n_{k}\right\}
_{i=2}^{m}\right)  \label{sum A} \\
&=&\left( tp\left( 1-p\right) \right) ^{-\frac{m(m-1)}{4}}\left( 2\pi
\right) ^{-\frac{m}{2}}\left( 1-\kappa \right) ^{-\frac{m(m-1)}{2}%
}\prod\limits_{i=1}^{m}\left( \xi _{i}-1\right) ^{-m}  \nonumber \\
&&\times \int\limits_{-\infty }^{\infty
}dy_{1}e^{-y_{1}^{2}/2}\int\limits_{y_{1}}^{\infty
}dy_{2}e^{-y_{1}^{2}/2}\cdots \int\limits_{y_{m-1}}^{\infty
}dy_{m}e^{-y_{1}^{2}/2}\prod\limits_{1\leq i<j\leq m}\left(
y_{j}-y_{i}\right)  \nonumber \\
&&\times \left[ 1+O\left( \frac{\left( \log t\right) ^{3}}{\sqrt{t}}\right) %
\right] .  \nonumber
\end{eqnarray}%
Combining (\ref{survival probability def}), (\ref{P^svw_T(X,X0)}) and (\ref%
{sum A}) we obtain
\begin{eqnarray}
\mathcal{P}_{t}\left( X^{0}\right) &=&\frac{1}{\left[ p(1-p)t\right]
^{m\left( m-1\right) /4}}\frac{I_{m,1/2}}{\left( 2\pi \right) ^{\frac{m}{2}%
}m!}  \nonumber  \label{P^swv_0} \\
&\times &\prod\limits_{i=1}^{m}\oint_{C_{\left\vert \xi _{i}\right\vert
=r>1}}\frac{d\xi _{i}}{2\pi \mathrm{i}\xi _{i}}\left( \frac{1-\kappa \xi _{i}%
}{1-\kappa }\right) ^{i-1}\prod\limits_{1\leq i<j\leq m}(\xi _{j}-\xi _{i})
\nonumber \\
&\times &\prod\limits_{i=1}^{m}\left( \xi _{i}-1\right) ^{-m}\xi
_{i}^{x_{m}^{0}-x_{i}^{0}+1}\left[ 1+O\left( \frac{\left( \log t\right) ^{3}%
}{\sqrt{t}}\right) \right] ,
\end{eqnarray}%
where $I_{m,1/2}$ is the Mehta integral defined in (\ref{Mehta}). Writing
the above product of integrals in determinant form and using the definition
of $I_{m,1/2}$ we arrive at the final result
\begin{eqnarray}
\mathcal{P}_{t}\left( X^{0}\right) &\simeq &\frac{2^{m}}{\pi ^{m/2}}\left[
p(1-p)t\right] ^{-\frac{m\left( m-1\right) }{4}}\left( 1-\kappa \right) ^{-%
\frac{m\left( m-1\right) }{2}}  \nonumber  \label{P^svw} \\
&&\times \prod\limits_{l=1}^{m}\Gamma \left( l/2+1\right) \det \left[ \left(
g_{i,j}(x_{m}^{0}-x_{i}^{0})\right) _{i,j=1}^{m}\right] .
\end{eqnarray}%
Here the function $g_{i,j}(x)$ is defined as follows
\begin{equation}
g_{i,j}(x)=\oint_{C_{0}}\frac{d\xi }{2\pi \mathrm{i}}\frac{\left( \kappa
+\kappa \xi -1\right) ^{i-1}\left( 1+\xi \right) ^{x}}{\xi ^{j}}.
\end{equation}

\paragraph{The case $\protect\kappa >1$}

Let $\kappa $ be a real number, $\kappa >1$. We return to the formulas (\ref%
{P^svw_T(X,X0)},\ref{A}). The crucial distinction from the case $\left\vert
\kappa \right\vert <1$ is that the presence of exponentially growing terms $%
\kappa ^{n_{i}}$ affects the range of values of $v_{1},\ldots ,v_{m}$, which
make the major contribution to the final sum of (\ref{P^svw_T(X,X0)}).
Indeed, we can write
\begin{equation}
\kappa ^{n_{i}}=e^{tv_{i}\log \kappa }\kappa ^{-x_{i}+x_{m}^{0}}.
\end{equation}%
Therefore, if we keep $v_{i}$ fixed, the sum over $x_{i}$ is rapidly
converging. At the same time the maximum of the $v_{i}$-dependent
exponential part of the r.h.s. of (\ref{P^svw_T(X,X0)}) is shifted due to
the appearance of the additional term $tv_{i}\log \kappa $. In a sense, the
roles of the variables $x_{2},\ldots ,x_{m}$ and $n_{2},\ldots ,n_{m}$ are
interchanged compared to the case $\kappa <1$.

Consequently, instead of summing over $n_{2},\ldots ,n_{m}$ and then over $%
x_{1},\ldots ,x_{m}$ we go to the variables \thinspace $v_{1},\ldots ,v_{m}$%
, (\ref{v_i kappa neq 1}), $\ $\ and $x_{2},\ldots ,x_{m}$ and evaluate the
sum over the latter first.%
\begin{eqnarray}
&&\sum_{x_{1}^{0}<x_{1}<x_{2}<\cdots <x_{m}}\sum_{\left\{ n_{2},\ldots
,n_{m}\right\} \in \mathbb{Z}_{\geq 0}^{m-1}}\mathcal{A}\left( \left\{ \xi
_{i},v_{i}\right\} _{i=1}^{m},\left\{ n_{k}\right\} _{i=2}^{m}\right)
\label{sum A1} \\
&=&\sum_{\left\{ v_{1},\ldots ,v_{m}\right\} \in t^{-1}\mathbb{Z}_{\geq
x_{1}^{0}}^{m}}\sum_{x_{2}=x_{1}+1}^{v_{2}t+x_{m}^{0}}\cdots
\sum_{x_{m}=x_{m-1}+1}^{v_{m}t+x_{m}^{0}}  \nonumber \\
&&\mathcal{A}\left( \left\{ \xi _{i},v_{i}\right\} _{i=1}^{m},\left\{
v_{i}-\left( x_{i}-x_{m}^{0}\right) /t\right\} _{i=2}^{m}\right)  \nonumber
\end{eqnarray}%
Collecting the factors of $\mathcal{A}\left( \left\{ \xi _{i},v_{i}\right\}
_{i=1}^{m},\left\{ v_{i}-\left( x_{i}-x_{m}^{0}\right) /t\right\}
_{i=2}^{m}\right) $ dependent on $x_{2},\ldots ,x_{m}$ we can evaluate the
sum over these variables
\begin{eqnarray}
&&\sum_{x_{2}=x_{1}+1}^{v_{2}t+x_{m}^{0}}\cdots
\sum_{x_{m}=x_{m-1}+1}^{v_{m}t+x_{m}^{0}}\prod_{i=2}^{m}\kappa
^{-x_{i}+x_{m}^{0}}\left(
\begin{array}{c}
tv_{i}-x_{i}+x_{m}^{0}+i-2 \\
tv_{i}-x_{i}+x_{m}^{0}%
\end{array}%
\right)  \nonumber \\
&=&\frac{\kappa ^{\left( m-1\right) \left( x_{m}^{0}-x_{1}\right) }}{\left(
\kappa -1\right) \cdots \left( \kappa ^{m-1}-1\right) }\prod_{i=2}^{m}\frac{%
\left( tv_{i}\right) ^{i-2}}{\left( i-2\right) !}\left( 1+O\left( \frac{1}{t}%
\right) \right) .  \label{sum over x}
\end{eqnarray}%
Here we extended the upper limit of all the summations to infinity, which
yields a correction of order of $\kappa ^{-v_{i}t}$, and we used Stirling's
formula to approximate the binomial coefficient%
\begin{equation}
\left(
\begin{array}{c}
i+n \\
n%
\end{array}%
\right) =\frac{n^{i}}{i!}\left( 1+O\left( \frac{1}{n}\right) \right) .
\label{binomial}
\end{equation}%
We also imply that the value of $v_{i}$ in the effective summation range is
finite and positive. Indeed, the range of summation over $v_{i}$ is defined
as above by the requirement that the exponential parts of $\mathcal{A}\left(
\left\{ \xi _{i},v_{i}\right\} _{i=1}^{m},\left\{ v_{i}-\left(
x_{i}-x_{m}^{0}\right) /t\right\} _{i=2}^{m}\right) $ are not too small.
Specifically, the exponentiated expressions are
\begin{equation}
\exp \left\{ t\left[ h_{i}\left( \frac{\left( 1-v_{i}\right) p}{\left(
1-p\right) v_{i}}\right) +v_{i}\log \kappa \right] \right\}  \label{exp1}
\end{equation}%
for $i=2,\ldots ,m$, and
\begin{equation}
\exp \left\{ t\left[ h_{1}\left( \frac{\left( 1-v_{i}\right) p}{\left(
1-p\right) v_{i}}\right) -\left( m-1\right) v_{1}\log \kappa \right]
\right\} ,  \label{exp2}
\end{equation}%
the term $t\left( 1-m\right) v_{1}\log \kappa =\log \left( \kappa ^{\left(
m-1\right) \left( x_{m}^{0}-x_{1}\right) }\right) $ in the latter coming
from the result of the summation over $x_{2},\ldots ,x_{m}$, (\ref{sum over
x}). For a real positive $\kappa $ the major contribution to the sums over $%
v_{1},\ldots ,v_{m}$ comes from the neighborhood of the maxima of the
exponentiated expressions
\begin{equation}
\left\vert v_{i}-u_{i}\right\vert <t^{-1/2}\log t
\end{equation}%
where the maxima $u_{i}$ are located at
\begin{equation}
u_{1}=\mathrm{u}(\kappa ^{1-m})
\end{equation}%
and
\begin{equation}
u_{i}=\mathrm{u}(\kappa )
\end{equation}%
for $i=2,\ldots ,m$ where
\begin{equation}
\mathrm{u}(x)=\frac{px}{1+\left( x-1\right) p}.  \label{statpoint}
\end{equation}%
Then we can follow the above procedure to evaluate the sums over $%
v_{1},\ldots ,v_{m}$. Going from the sum to an integral over the variables%
\begin{equation}
y_{i}=\sqrt{t}\frac{v_{i}-u_{i}}{\sqrt{u_{i}\left( 1-u_{i}\right) }}.
\end{equation}%
we arrive at the integral expression
\begin{eqnarray}
&&\sum_{\left\{ v_{1},\ldots ,v_{m}\right\} \in t^{-1}\mathbb{Z}_{\geq
x_{1}^{0}}^{m}}\sum_{x_{2}=x_{1}+1}^{v_{2}t+x_{m}^{0}}\cdots
\sum_{x_{m}=x_{m-1}+1}^{v_{m}t+x_{m}^{0}}  \nonumber \\
&&\mathcal{A}\left( \left\{ \xi _{i},v_{i}\right\} _{i=1}^{m},\left\{
v_{i}-\left( x_{i}-x_{m}^{0}\right) /t\right\} _{i=2}^{m}\right)  \nonumber
\\
&=&\left( 1-p+\kappa p\right) ^{\left( m-1\right) t}\left( 1-p+\kappa
^{1-m}p\right) ^{t}  \nonumber \\
&&\times \frac{\left( u-u_{1}\right) ^{m-1}\left( 1-u\right) ^{\frac{\left(
m-1\right) \left( m-2\right) }{2}}}{u^{m\left( m-1\right) /2}u_{1}^{m-1}}%
\left( \frac{p}{1-p}\right) ^{\frac{m(m-1)}{2}}  \nonumber \\
&&\times \frac{\prod\limits_{k=1}^{m}\left( \xi _{k}-\frac{\left( 1-u\right)
p}{\left( 1-p\right) u}\right) ^{1-m}\left( \xi _{k}-\frac{\left(
1-u_{1}\right) p}{\left( 1-p\right) u_{1}}\right) ^{-1}}{\left( 2\pi \right)
^{m/2}\prod\limits_{i=2}^{m}\left[ \left( i-2\right) !\left( \kappa
^{i-1}-1\right) \right] }  \nonumber \\
&&\times \int\limits_{-\infty }^{+\infty }dy_{1}e^{-y_{1}^{2}/2}\cdots
\int\limits_{-\infty }^{+\infty }dy_{m}e^{-y_{m}^{2}/2}\prod\limits_{2\leq
l<j\leq m}^{m}\left( y_{j}-y_{l}\right)  \nonumber \\
&&\times \prod\limits_{s=2}^{m}\left( y_{s}+\frac{u\sqrt{t}}{\sqrt{u(1-u)}}%
\right) ^{s-2}\left( 1+O\left( \frac{1}{t}\right) \right) .
\label{sum kappa>1}
\end{eqnarray}%
Note that we keep the leading terms of the Taylor expansion in $\left(
v_{i}-u_{i}\right) $, $i=1,\ldots ,m$, everywhere under the integral except
the product in the last line, where we keep the terms of two subsequent
orders. The reason for the latter is that the leading terms of the
multipliers cancel due to the antisymmetry of the rest of the expression in $%
v_{2},\ldots ,v_{m}$. Therefore, what contributes is the antisymmetric part
of this line, that is $\prod\nolimits_{2\leq l<j\leq m}^{m}\left(
y_{l}-y_{j}\right) /\left( m-1\right) !$, which contains only the terms of
the same order. Inserting it, we again arrive at the Mehta integral $%
I_{m-1,1}$ over $(m-1)$, variables $y_{2},\ldots ,y_{m}$, while the integral
over $y_{1}$ decouple being just the Laplace integral. After substitution of
the explicit form of $u$ and $u_{1}$ the r.h.s. of (\ref{sum kappa>1})
becomes
\begin{eqnarray}
&&\frac{\left( 1-p+\kappa p\right) ^{\left( m-1\right) t}\left( 1-p+\kappa
^{1-m}p\right) ^{t}\left( \kappa ^{m}-1\right) ^{m-1}\kappa ^{-\frac{m\left(
m-1\right) }{2}}}{\prod\limits_{i=2}^{m}\left[ \left( i-1\right) !\left(
\kappa ^{i-1}-1\right) \right] }  \nonumber \\
&&\times I_{m-1,1}\left( 2\pi \right) ^{-\frac{m-1}{2}}\prod%
\limits_{k=1}^{m}\left( \xi _{k}-\frac{1}{\kappa }\right) ^{1-m}\left( \xi
_{k}-\kappa ^{m-1}\right) ^{-1}.
\end{eqnarray}%
This formula together with (\ref{survival probability def}), (\ref%
{P^svw_T(X,X0)}) and (\ref{Mehta}) yields
\begin{eqnarray}
\hspace{-0.5cm}\mathcal{P}_{t}\left( X^{0}\right) &=&\left( 1-p+\kappa
p\right) ^{\left( m-1\right) t}\left( 1-p+\kappa ^{1-m}p\right) ^{t}
\nonumber \\
&\times &\frac{\left( \kappa ^{m}-1\right) ^{m-1}\left( -1\right) ^{\frac{%
m\left( m-1\right) }{2}}}{\prod\limits_{i=1}^{m-1}\left( \kappa
^{i}-1\right) }\prod\limits_{k=1}^{m}\oint_{C_{r>\kappa ^{m-1}}}\frac{d\xi
_{k}}{2\pi \mathrm{i}}  \label{P_t(X^0) kappa>1 (2)} \\
&\times &\frac{\xi _{k}^{x_{m}^{0}-x_{k}^{0}}}{\left( \xi _{k}-\kappa
^{-1}\right) ^{m-k}\left( \xi _{k}-\kappa ^{m-1}\right) }\prod\limits_{1\leq
i<j\leq m}(\xi _{j}-\xi _{i}).  \nonumber
\end{eqnarray}%
The integration over $\xi _{k}$, $k=1,\ldots ,m$ is performed over a circle
encircling the singularities of the expression under the integral, i.e. $%
\left\vert \xi _{k}\right\vert >\kappa ^{m-1}$. First we integrate over $\xi
_{m}$, then over $\xi _{m-1}$, $\xi _{m-2},\cdots ,\xi _{1}$. It turns out
that, if we integrate in this order, the expression under the integral being
evaluated will contain only one simple pole at each step. As a result we
arrive at the simple one term expression
\begin{eqnarray}
&&\prod_{k=1}^{m}\oint_{C_{r>\kappa ^{m-1}}}\frac{d\xi _{k}}{2\pi \mathrm{i}}%
\frac{\xi _{k}^{x_{m}^{0}-x_{k}^{0}}}{\left( \xi _{k}-\kappa ^{-1}\right)
^{m-k}\left( \xi _{k}-\kappa ^{m-1}\right) }\prod\limits_{1\leq i<j\leq
m}(\xi _{j}-\xi _{i}) \\
&=&\left( -1\right) ^{\frac{m\left( m-1\right) }{2}}\kappa
^{\sum_{k=1}^{m-1}\left( x_{k}^{0}-x_{m}^{0}\right) }  \nonumber
\end{eqnarray}%
This formula together with (\ref{P_t(X^0) kappa>1 (2)}) results in (\ref%
{P_t(X^0) kappa>1}).

\begin{remark}
An extension of the approximation technique used to complex values of $%
\kappa $ is problematic for $\left\vert \kappa \right\vert >1$. The reason
is that in this case the critical points $u_{1},\ldots ,u_{m}$ are away from
the real axis. It follows then that their contribution can be exponentially
smaller the other corrections appearing.
\end{remark}

\subsubsection{The limiting case $\protect\kappa \rightarrow 1$}

Now we consider the limiting case
\begin{equation}
t\rightarrow \infty ,\kappa \rightarrow 1,\left( 1-\kappa \right) \sqrt{t}%
=const.
\end{equation}%
We start with the formula (\ref{int}) and expand the term $(1-\kappa
z_{i})^{-i+1}$ into its Taylor series,%
\begin{eqnarray}
&&P_{t}\left( X,X^{0}\right) =\sum\limits_{\left\{ n_{i}\right\} \in \mathbb{%
Z}_{\geq 0}^{m}}\prod\limits_{i=2}^{m}\kappa ^{n_{i}}\left(
\begin{array}{c}
i+n_{i}-2 \\
n_{i}%
\end{array}%
\right) \\
&&\times \det \left( \oint \frac{dz_{i}}{2\pi \mathrm{i}z_{i}}\left( 1-p+%
\frac{p}{z_{i}}\right) ^{t}z_{i}^{x_{i}+n_{i}-x_{j}^{0}}(1-\kappa
z_{i})^{j-1}\right) _{1\leq i,j\leq m}.  \nonumber
\end{eqnarray}%
The integral in the determinant can\ be evaluated in the saddle point
approximation, the analysis being similar to the one above, (\ref{f_i(z)})-(%
\ref{P_T(X,X^0)_1}), with the same function $h_{i}(z)$, (\ref{f_i(z)})
except that $v_{i}$ now depends on $n_{i}$:
\begin{equation}
v_{i}=\frac{x_{i}+n_{i}-x_{m}^{0}}{t}.
\end{equation}%
What is special about the limit $\kappa \rightarrow 1$ is that the saddle
point can coincide with a zero of the factor $(1-\kappa z)^{j}$ within the
effective range of the summation over $v_{i}$. Therefore, instead of
expanding this term into a Taylor series, we leave it in the integral as is,
while the rest can be expanded around the saddle point as usual. Then we use
the following formula for Hermite polynomials (see \cite{Gradshtein Ryzhik},
formula 3.462.4),
\begin{equation}
\int\limits_{-\infty }^{\infty }e^{-x^{2}}\left( x-\beta \right) ^{n}dx=%
\sqrt{\pi }\left( \frac{\mathrm{i}}{2}\right) ^{n}H_{n}\left( \mathrm{i}%
\beta \right) .
\end{equation}%
As a result we obtain
\begin{eqnarray}
\lefteqn{\oint \frac{dz}{2\pi \mathrm{i}z}%
e^{th_{i}(z)}z^{x_{m}^{0}-x_{j}^{0}}(1-\kappa z)^{j-1}=e^{th_{i}\left(
z^{\ast }\right) }\left( z^{\ast }\right) ^{x_{m}^{0}-x_{j}^{0}-1}}
\label{integral1} \\
&&\times \int_{-\infty }^{\infty }\frac{d\xi }{2\pi }e^{-\frac{1}{2}%
t|h_{i}^{\prime \prime }\left( z^{\ast }\right) |\xi ^{2}}(1-\kappa \left(
z^{\ast }+\mathrm{i}\xi \right) )^{j-1}  \nonumber \\
&=&\frac{e^{th_{i}\left( \frac{\left( 1-v_{i}\right) p}{\left( 1-p\right)
v_{i}}\right) }}{\sqrt{\pi }\left( 2t\right) ^{j/2}}H_{j-1}\left( \sqrt{%
\frac{tv_{i}\left( 1-v_{i}\right) }{2}}\left( \frac{1}{\kappa }\frac{\left(
1-p\right) v_{i}}{\left( 1-v_{i}\right) p}-1\right) \right)  \nonumber \\
&&\times \kappa ^{j-1}\left( \frac{p}{1-p}\right)
^{j+x_{m}^{0}-x_{j}^{0}}\!\!\frac{\left( 1-v_{i}\right)
^{x_{m}^{0}-x_{j}^{0}-j/2}}{v_{i}^{x_{m}^{0}-x_{j}^{0}-3j/2}}\left(
1+O\left( \frac{1}{t}\right) \right) .  \nonumber
\end{eqnarray}%
Next, we argue that the dominant range of the summation over $X$ and $%
\left\{ n_{i}\right\} _{i=1}^m$ is the domain
\begin{equation}
pt-\sqrt{t}\log t\leq x_{i}+n_{i}\leq pt+\sqrt{t}\log t,  \label{range}
\end{equation}%
where $x_{i}$ varies within the range%
\begin{equation}
x_{1}^{0}\leq x_{1}\leq \cdots \leq x_{m}\leq t,
\end{equation}%
and
\begin{equation}
0\leq n_{i}<\infty .  \label{n_i range}
\end{equation}%
To this end, consider the integral (\ref{integral1}) for some particular $i$
and $j$. After expanding $(1-\kappa z_{i})^{j}$ into a binomial sum, it
becomes a finite sum of terms like%
\begin{equation}
\left( 1-p\right)
^{t-(n_{i}+x_{i}-x_{j}^{0}+k)}p^{n_{i}+x_{i}-x_{j}^{0}+k}\left(
\begin{array}{c}
t \\
n_{i}+x_{i}-x_{j}^{0}+k%
\end{array}%
\right) ,
\end{equation}%
where $k$ is a finite integer, $0\leq k\leq j$. Beyond the range (\ref{range}%
) this can be estimated using Stirling's formula to be $O(t^{-1/2}e^{-\frac{%
\left( \log t\right) ^{2}}{2p\left( 1-p\right) }})$. The summation over $%
x_{i}$, which includes at most $t$ nonzero terms, multiplies this estimate
by a factor of $t$. Finally the summation over $n_{i}$ yields an additional
factor $\left( 1-\kappa \right) ^{-i}$, with the result that the order of
the contribution from outside the domain (\ref{range}) is
\begin{equation}
O\left( \left( 1-\kappa \right) ^{-i}t^{1/2}e^{-\frac{\left( \log t\right)
^{2}}{2p\left( 1-p\right) }}\right) .  \label{correction}
\end{equation}%
Below, the leading term of the sum of interest will be shown to decay at
most as a power of $t$. Therefore, when $\kappa $ is such that $\left(
1-\kappa \right) =O(t^{-s})$ with any fixed $s>0$, the term (\ref{correction}%
) is asymptotically negligible.

One can approximate (\ref{integral1}) using the Taylor formula, which yields
\begin{eqnarray}
&&\det \left[ H_{j-1}\left( \sqrt{\frac{tp\left( 1-p\right) }{2}}\left(
\frac{v_{i}-p}{p\left( 1-p\right) }-1+\frac{1}{\kappa }\right) \right) %
\right] _{i,j=1}^{m}  \nonumber \\
&&\qquad \times \frac{\prod\limits_{i=1}^{m}e^{-t\frac{\left( v_{i}-p\right)
^{2}}{2p\left( 1-p\right) }}}{\pi ^{\frac{m}{2}}\left[ 2tp\left( p-1\right) %
\right] ^{\frac{m(m+1)}{4}}}\left( 1+O\left( \frac{\left( \log t\right) ^{3}%
}{\sqrt{t}}\right) \right) .
\end{eqnarray}%
The determinant can be simplified by adding to every line a multiple of the
lines below it, such that all the terms of the Hermite polynomials except
the highest cancel:
\begin{equation}
\det \left[ H_{j-1}\left( a_{i}\right) \right] _{i,j=1}^{m}=\left( -2\right)
^{\frac{m\left( m-1\right) }{2}}\Delta (a_{1},\ldots ,a_{m})
\end{equation}%
Thus, the survival probability takes the following form%
\begin{eqnarray}
\mathcal{P}_{t}\left( X^{0}\right)
&=&\sum\limits_{\{x_{i}\}_{i=1}^{m}}\sum\limits_{\{n_{i}\}_{i=2}^{m}}\frac{%
(-1)^{\frac{m\left( m-1\right) }{2}}}{\left( 2\pi \right) ^{\frac{m}{2}}t^{%
\frac{m(m+1)}{2}}\left[ p\left( p-1\right) \right] ^{\frac{m^{2}}{2}}}
\nonumber  \label{P_t^svw} \\
&\times &\prod\limits_{i=2}^{m}\kappa ^{n_{i}}e^{-\frac{\left(
x_{i}+n_{i}-x_{m}^{0}-pt\right) ^{2}}{2p(1-p)t}}\left(
\begin{array}{c}
i+n_{i}-2 \\
n_{i}%
\end{array}%
\right) \\
&\times &\Delta (x_{1},x_{2}+n_{2},\ldots ,x_{m}+n_{m})\left[ 1+O\left(
\frac{\left( \log t\right) ^{3}}{\sqrt{t}}\right) \right] ,  \nonumber
\end{eqnarray}%
where the summation is over the domains of $\{x_{i}\}_{i=1}^{m}$ and $%
\{n_{i}\}_{i=2}^{m}$ defined by the inequalities (\ref{range})-(\ref{n_i
range}). Due to the presence of the Gaussian factor $\exp(-t(v_{i}-p)^{2}/(
2p( 1-p) t))$ the sum over $n_{i}$ converges uniformly in $x_{i}$. Therefore
we can interchange the order of summations over $x_{i}$ and over $n_{i}$.
This allows us to apply Lemma \ref{sum-to-integral} first to the variables $%
x_{i}$ and then also to the variables $n_{i}$. As the characteristic scale
of $n_{i}$ is of order $\sqrt{t}$ we can use the approximation (\ref%
{binomial}) for the binomial coefficient, where the correction term yields
an error of the order $O\left( t^{-1/2}\right) $ in the final result. To
write down the final limiting formula as $t\rightarrow \infty $ for $%
\mathcal{P}\left( X^{0}\right) $ we introduce the rescaled variables
\begin{eqnarray}
u_{i} &=&\frac{\left( x_{i}-pt\right) }{\sqrt{tp\left( 1-p\right) }}, \\
\nu _{i} &=&\frac{n_{i}}{\sqrt{tp\left( 1-p\right) }}.
\end{eqnarray}%
and the transition parameter $\alpha $,$\ $(\ref{alpha}), which is constant
in the limit under consideration. The formula (\ref{P_t^svw}) then takes the
form (\ref{result transition}) , where $f_{m}\left( \alpha \right) $ is
given by (\ref{f_m(alpha)}).

\section{Asymptotic behaviour of $f_{m}\left( \protect\alpha \right) \label%
{f_m(alpha)_section}$}

In this section we evaluate the limiting behaviour of $f_{m}(\alpha )$ for $%
\alpha \rightarrow \infty $ and $\alpha \rightarrow -\infty $ and its value
at $\alpha =0$ . In the latter case it is just the probability normalization
of the TASEP, so $f_{m}\left( 0\right) $ must be equal $1$. The limit $%
\alpha \rightarrow -\infty $ has no probabilistic meaning, but it can be
considered a particular limit of the generating function of the rescaled
particle current in the TASEP: see Section $2$. Let us introduce the
notation
\begin{eqnarray}
J_{m}\left( \alpha \right) &=&\int\limits_{-\infty }^{\infty
}du_{1}\int\limits_{u_{1}}^{\infty }du_{2}\cdots
\!\!\int\limits_{u_{m-1}}^{\infty }du_{m}\int\limits_{0}^{\infty }d\nu
_{2}\cdots \!\int\limits_{0}^{\infty }d\nu _{m} \\
&\times &e^{-\frac{1}{2}u_{1}^{2}}\prod\limits_{i=2}^{m}\nu _{i}^{i-2}e^{-%
\frac{1}{2}\left( u_{i}+\nu _{i}\right) ^{2}-\alpha \nu _{i}}\Delta \left(
u_{1},\nu _{2}+u_{2},\ldots ,\nu _{m}+u_{m}\right) .  \nonumber
\end{eqnarray}%
for the multiple integral entering into the expression of $f_{m}\left(
\alpha \right) $. Then we have
\begin{equation}
f_{m}\left( \alpha \right) =\frac{\left( -1\right) ^{\frac{m\left(
m-1\right) }{2}}}{\left( 2\pi \right) ^{\frac{m}{2}}2!\cdots \left(
m-2\right) !}J_{m}\left( \alpha \right) .  \label{f_m(alpha) via J_m}
\end{equation}%
The form of $J_{m}\left( \alpha \right) $ is reminiscent of the multiple
integrals which appear in the theory of Gaussian random matrix ensembles.
The following three lemmas show that in the three limiting cases $%
J_{m}\left( \alpha \right) $ can be explicitly evaluated in the form of
Mehta integrals.

\begin{lemma}
\begin{equation}
\lim_{\alpha \rightarrow \infty }\alpha ^{\frac{m\left( m-1\right) }{2}%
}J_{m}\left( \alpha \right) =I_{m,1/2}\frac{\left( -1\right) ^{\frac{m\left(
m-1\right) }{2}}2!\cdots \left( m-2\right) }{m!}  \label{alpha->0}
\end{equation}%
where $I_{m,1/2}$ is the Mehta integral defined in (\ref{Mehta}).
\end{lemma}

\begin{proof}
Let us make\ a variable change under the integral (\ref{f_m(alpha)})
introducing new integration variables
\begin{eqnarray}
\varphi _{1} &=&u_{1}, \\
\varphi _{i} &=&\nu _{i}+u_{i},i=2,\ldots ,m, \\
\mu _{i-1} &=&\alpha \nu _{i},i=2,\ldots ,m.
\end{eqnarray}%
In the new variables the integral (\ref{f_m(alpha)}) can be written as%
\begin{equation}
J_{m}\left( \alpha \right) =\frac{1}{\alpha ^{m(m-1)/2}}\prod%
\limits_{i=1}^{m-1}\int\limits_{0}^{\infty }d\mu _{i}\mu _{i}^{i-1}e^{-\mu
_{i}}g(\mu _{1},\ldots ,\mu _{m-1};\alpha ),  \label{f(alpha)_1}
\end{equation}%
where we introduce the notation
\begin{eqnarray}
\lefteqn{g(\mu _{1},\ldots ,\mu _{m-1};\alpha )}  \nonumber \\
&=&\int\limits_{-\infty }^{\infty }d\varphi _{1}\int\limits_{\varphi _{1}+%
\frac{\mu _{1}}{\alpha }}^{\infty }d\varphi _{2}\int\limits_{\varphi _{2}+%
\frac{\mu _{2}-\mu _{1}}{\alpha }}^{\infty }\!d\varphi _{3} \\
\cdots &&\!\!\!\int\limits_{\varphi _{m-1}+\frac{\mu _{m-1}-\mu _{m-2}}{%
\alpha }}^{\infty }d\varphi _{m}e^{-\frac{1}{2}\left( \varphi
_{m}^{2}+\cdots +\varphi _{1}^{2}\right) }\Delta \left( \varphi _{1},\ldots
,\varphi _{m}\right) .
\end{eqnarray}%
The function $g(\mu _{1},\ldots ,\mu _{m-1};\alpha )$ is bounded uniformly
in $\alpha \in \mathbb{R}$.
\begin{equation}
\left\vert g(\mu _{1},\ldots ,\mu _{m-1};\alpha )\right\vert \leq I_{m,1/2},
\end{equation}%
which can be shown by replacing the Vandermonde determinant under the
integral by its absolute value and extending the lower integration limits to
minus infinity. Therefore the function under the integral in (\ref%
{f(alpha)_1}) is uniformly bounded and integrable. By the dominating
convergence theorem one can interchange the limit $\alpha \rightarrow \infty
$ and integration. Then, for the function $g(\mu _{2},\ldots ,\mu
_{m};\alpha )$ we have
\begin{equation}
\lim_{\alpha \rightarrow \infty }g(\mu _{2},\ldots ,\mu _{m};\alpha
)=I_{m,1/2}\frac{\left( -1\right) ^{\frac{m\left( m-1\right) }{2}}}{m!}.
\end{equation}%
Remarkably the limiting value does not depend on the variables \newline
$\left\{ \mu _{1},\ldots ,\mu _{m-1}\right\} $. Therefore the integration in
(\ref{f(alpha)_1}) can be performed independently for each $i=2,\ldots ,m$,
each resulting in $\left( i-1\right) !$. This yields (\ref{alpha->0}).
\end{proof}

\begin{lemma}
\begin{equation}
J_{m}\left( 0\right) =\frac{\left( -1\right) ^{\frac{m(m-1)}{2}}}{\left(
m-1\right) !m!}I_{m,1}
\end{equation}
\end{lemma}

\begin{proof}
Let us make the variable change
\begin{eqnarray}
\chi _{1} &=&u_{1}  \nonumber \\
\chi _{i} &=&\nu _{i}+u_{i},i=2,\ldots ,m.  \label{variable change}
\end{eqnarray}%
Then the integral takes the form%
\begin{eqnarray}
&&J_{m}\left( 0\right) =\int\limits_{-\infty }^{\infty }d\chi
_{1}\int\limits_{\chi _{1}}^{\infty }du_{2}\int\limits_{u_{2}}^{\infty
}du_{3}\cdots \int\limits_{u_{m-1}}^{\infty
}du_{m}\int\limits_{u_{2}}^{\infty }d\chi _{2}  \label{form1} \\
&&\cdots \int\limits_{u_{m}}^{\infty }d\chi _{m}e^{-\frac{1}{2}\chi
_{1}^{2}}\prod\limits_{i=2}^{m}\left( \chi _{i}-u_{i}\right) ^{i-2}e^{-\frac{%
1}{2}\chi _{i}^{2}}\Delta \left( \chi _{1},\ldots ,\chi _{m}\right) .
\nonumber
\end{eqnarray}%
The integrals over $u_{i}$, for $i=1,\ldots ,m$, can be evaluated step by
step. First, for $i=m$, we have%
\begin{eqnarray}
&&\int\limits_{u_{m-1}}^{\infty }du_{m}\int\limits_{u_{m}}^{\infty }d\chi
_{m}e^{-\frac{1}{2}\chi _{m}^{2}}\left( \chi _{m}-u_{m}\right) ^{m-2}\Delta
\left( \chi _{1},\ldots ,\chi _{m}\right) \\
&=&\frac{1}{m-1}\int\limits_{u_{m-1}}^{\infty }d\chi _{m}e^{-\frac{1}{2}\chi
_{m}^{2}}\left( \chi _{m}-u_{m-1}\right) ^{m-1}\Delta \left( \chi
_{1},\ldots ,\chi _{m}\right) ,  \nonumber
\end{eqnarray}%
which is done by changing the integration order. In the next step, the
integral over $u_{m-1}$ can be calculated by parts:
\begin{eqnarray}
&&\int\limits_{u_{m-2}}^{\infty }du_{m-1}\int\limits_{u_{m-1}}^{\infty
}d\chi _{m-1}\int\limits_{u_{m-1}}^{\infty }d\chi _{m}e^{-\frac{1}{2}\chi
_{m-1}^{2}-\frac{1}{2}\chi _{m}^{2}}  \nonumber \\
&&\times \left( \chi _{m-1}-u_{m-1}\right) ^{m-3}\left( \chi
_{m}-u_{m-1}\right) ^{m-1}\Delta \left( \chi _{1},\ldots ,\chi _{m}\right)
\nonumber \\
&=&\frac{1}{m-2}[\int\limits_{u_{m-2}}^{\infty }d\chi
_{m-1}\int\limits_{u_{m-2}}^{\infty }d\chi _{m}e^{-\frac{1}{2}\chi
_{m-1}^{2}-\frac{1}{2}\chi _{m}^{2}} \\
&&\times \left( \chi _{m-1}-u_{m-2}\right) ^{m-2}\left( \chi
_{m}-u_{m-2}\right) ^{m-1}\Delta \left( \chi _{2},\ldots ,\chi _{m}\right)
\nonumber \\
&&-\left( m-1\right) \int\limits_{u_{m-2}}^{\infty
}du_{m-1}\int\limits_{u_{m-1}}^{\infty }d\chi
_{m-1}\int\limits_{u_{m-1}}^{\infty }d\chi _{m}e^{-\frac{1}{2}\chi
_{m-1}^{2}-\frac{1}{2}\chi _{m}^{2}}  \nonumber \\
&&\times \left( \chi _{m-1}-u_{m-1}\right) ^{m-2}\left( \chi
_{m}-u_{m-1}\right) ^{m-2}\Delta \left( \chi _{1},\ldots ,\chi _{m}\right) ].
\nonumber
\end{eqnarray}%
The second term cancels because of the antisymmetry of the Vandermonde
determinant with respect to interchange of $\chi _{m}$ and $\chi _{m-1}$.
Iterating this procedure we remove $\left( m-1\right) $ integrals in the
variables $u_{2},\ldots ,u_{m}$ .
\begin{eqnarray}
&&\frac{1}{\left( m-1\right) !}\int\limits_{-\infty }^{\infty }d\chi
_{1}\int\limits_{\chi _{1}}^{\infty }d\chi _{2}\cdots \int\limits_{\chi
_{1}}^{\infty }d\chi _{m} \\
&&\times e^{-\chi _{1}^{2}}\prod\limits_{i=2}^{m}\left( \chi _{i}-\chi
_{1}\right) ^{i-1}e^{-\frac{1}{2}\chi _{i}^{2}}\Delta \left( \chi
_{1},\ldots ,\chi _{m}\right)  \nonumber
\end{eqnarray}%
A symmetrization of the expression under the integral in the variables
$\chi_{i}$, $i=2,\ldots,m$, yields another Vandermonde determinant. Thus
\begin{eqnarray}
J_{m}\left( 0\right) &=&\frac{\left( -1\right) ^{\frac{m(m-1)}{2}}}{\left(
\left( m-1\right) !\right) ^{2}}\int\limits_{-\infty }^{\infty }d\chi
_{1}\int\limits_{\chi _{1}}^{\infty }d\chi _{2} \\
&&\qquad \cdots \int\limits_{\chi _{1}}^{\infty }d\chi _{m}e^{-\frac{1}{2}%
\left( \chi _{1}^{2}+\cdots +\chi _{m}^{2}\right) }\left\vert \Delta \left(
\chi _{1},\ldots ,\chi _{m}\right) \right\vert ^{2}.  \nonumber
\end{eqnarray}%
Finally we add this integral to the $\left( m-1\right) $ similar ones,
obtained by interchanging $\chi _{1}$ with each of $\chi _{2},\ldots ,\chi
_{m}$, and divide the sum by $m$.
\begin{eqnarray}
J_{m}\left( 0\right) &=&\frac{\left( -1\right) ^{\frac{m(m-1)}{2}}}{\left(
m-1\right) !m!}\int\limits_{-\infty }^{\infty }d\chi _{1} \\
&&\cdots \int\limits_{-\infty }^{\infty }d\chi _{m}e^{-\frac{1}{2}\left(
\chi _{1}^{2}+\cdots +\chi _{m}^{2}\right) }\left\vert \Delta \left( \chi
_{1},\ldots ,\chi _{m}\right) \right\vert ^{2}  \nonumber
\end{eqnarray}%
This gives us the stated result.
\end{proof}

\begin{lemma}
\begin{equation}
\lim_{\alpha \rightarrow -\infty }e^{-\alpha ^{2}m(m-1)/2}J_{m}\left( \alpha
\right) =I_{m-1,1}\frac{\sqrt{2\pi }\left( -1\right) ^{\frac{m(m-1)}{2}}}{%
\left( \left( m-1\right) !\right) ^{2}}m^{m-1}.  \label{alha->-infty (lemma)}
\end{equation}
\end{lemma}

\begin{proof}
We start from the integral in (\ref{f_m(alpha)}) and make a change of
variables as follows,
\begin{eqnarray}
x_{i} &=&\nu _{i}+\alpha +u_{1},\quad i=1,\ldots ,m-1 \\
s_{i} &=&\left\vert \alpha \right\vert \left( u_{i}-u_{1}\right) ,\quad
i=1,\ldots ,m-1 \\
s_{1} &=&u_{1}-\alpha \left( m-1\right) ,
\end{eqnarray}%
which yields the following integral expression,
\begin{eqnarray}
&&J_{m}\left( \alpha \right) =\frac{e^{\alpha ^{2}\frac{m(m-1)}{2}}}{%
\left\vert \alpha \right\vert ^{m-1}}\int\limits_{-\infty }^{\infty
}ds_{1}e^{-\frac{1}{2}s_{1}^{2}}\int\limits_{0}^{\infty
}ds_{2}e^{-s_{2}}\int\limits_{s_{2}}^{\infty }ds_{3}e^{-s_{3}}\cdots \\
&&\times \int\limits_{s_{m-1}}^{\infty
}ds_{m}e^{-s_{m}}\int\limits_{s_{1}+\alpha m}^{\infty }dx_{2}\cdots
\int\limits_{s_{1}+\alpha m}^{\infty }dx_{m}\prod\limits_{i=2}^{m}\left(
x_{i}-s_{1}-\alpha m\right) ^{i-2}  \nonumber \\
&&\times \prod_{i=2}^{m}e^{-\frac{1}{2}\left[ x_{i}^{2}+\frac{1}{\left\vert
\alpha \right\vert }\left( 2s_{i}x_{i}+\frac{s_{i}^{2}}{|\alpha |}\right) %
\right] }\Delta \left( s_{1}+\alpha m,x_{2}+\frac{s_{2}}{\left\vert \alpha
\right\vert }\ldots ,x_{m}+\frac{s_{m}}{\left\vert \alpha \right\vert }%
\right) .  \nonumber
\end{eqnarray}%
Due to the presence of the Gaussian and exponential terms, the main
contribution to the integral comes from finite values of $s_{1},\ldots
,s_{m} $ and $x_{2},\ldots ,x_{m}$. Therefore, up to corrections of order of
$O(1/\left\vert \alpha \right\vert )$, we can neglect the terms divided by $%
\left\vert \alpha \right\vert $, and extend the lower limits of integration
over $x_{2},\ldots ,x_{m}$ to $-\infty $. The the integrals over $%
s_{2},\ldots ,s_{m}$ decouple, and we can evaluate them to $1/\left(
m-1\right) !$. The Vandermonde determinant becomes antisymmetric with
respect to permutations of the variables $x_{2},\ldots ,x_{m}$. As the
integration is over the symmetric domain, we can leave only the
antisymmetric part of the rest of the expression. The product $%
\prod_{i=2}^{m}\left( x_{i}-s_{1}-\frac{\alpha m}{2}\right) ^{i-2}$ then
results in
\begin{equation}
\left( -1\right) ^{\frac{\left( m-1\right) \left( m-2\right) }{2}}\frac{%
\Delta \left( x_{2}\ldots ,x_{m}\right) }{(m-1)!}.
\end{equation}%
After collecting the leading order terms from the first argument of
\begin{equation}
\Delta \left( s_{1}+\alpha m,x_{2}\ldots ,x_{m}\right) \simeq \left( \alpha
m\right) ^{m-1}\Delta \left( x_{2}\ldots ,x_{m}\right) ,
\end{equation}%
the integral over $s_{1}$ decouples as well, and yields $\sqrt{\pi }$. We
finally obtain%
\begin{eqnarray}
&&J_{m}\left( \alpha \right) \simeq \frac{e^{\frac{\alpha ^{2}m(m-1)}{2}}%
\sqrt{2\pi }\left( -1\right) ^{\frac{\left( m-1\right) (m-2)}{2}}}{%
\left\vert \alpha \right\vert ^{m-1}\left( \left( m-1\right) !\right) ^{2}}%
\left( \alpha m\right) ^{m-1}  \nonumber \\
&&\qquad \times \int\limits_{-\infty }^{\infty }dx_{2}\cdots
\int\limits_{-\infty }^{\infty }dx_{m}e^{-\frac{1}{2}\sum%
\limits_{i=2}^{m}x_{i}^{2}}\left\vert \Delta \left( x_{2},\ldots
,x_{m}\right) \right\vert ^{2}.
\end{eqnarray}%
Using the definition of the Mehta integral, (\ref{Mehta}) we obtain (\ref%
{alha->-infty (lemma)}).
\end{proof}

The above lemmas, the formula for the Mehta integral (\ref{Mehta}) and the
definition (\ref{f_m(alpha) via J_m}) of $f_{m}\left( \alpha \right) $
establish the results (\ref{alpha->+infty})-(\ref{alpha->-infty}).

\section{Discussion of the results and conclusion}

The first result obtained in this paper is an expression for the survival
probability for $m$ walkers in the SVW model. At the fixed parameter $\kappa
<1$ and $t\rightarrow \infty $ the leading asymptotics $t^{-m(m-1)/4}$
coincides with that for the usual VW \cite{Fisher}. This result is
intuitively clear. In the case of the VW it is obtained by considering the
evolution of independent noninteracting particles and reducing the number of
its outcomes by dropping those realizations where the crossings of particle
time space trajectories occur. Then, the asymptotics of the survival
probability for two walkers $t^{-1/2}$ follows directly from the diffusion
law and the method of images. For $m$ walkers, the method of images involves
$m(m-1)/2$ reflecting wall planes. Each wall brings with it a factor $%
t^{-1/2}$ providing total contribution $t^{-m(m-1)/4}$. For any fixed $%
\kappa <1$ the events of crossing occur with a finite, though changed,
probability, so that this argument still holds. The survival probability
changes by a constant factor dependent on $\kappa $, leaving Fisher's law
unchanged. As $\kappa $ approaches one this factor diverges, which indicates
a qualitative change of behaviour of the survival probability. As the
crossings become less probable, it finally saturates to the TASEP
normalization constant. The transition between the two regions takes place
on the scale $(1-\kappa )\sqrt{t}=const$, where the effect of the diffusive
spreading of particles becomes comparable to the one caused by SVW
interaction. In this case, the survival probability is expressed by the
scaling function $f_{m}(\alpha )$, which has the SVW and TASEP asymptotics
as limiting cases.

The formulas for the survival probability in SVW can be reinterpreted in
terms of the moment generating function of the time integrated particle
current $Y_{t}$ in TASEP, which, in turn, can be used to construct the
distribution of the same quantity. The cases of generic $\kappa <1$ and $%
\kappa >1$ correspond to the tails of the probability distribution of $Y_{t}$
at the large deviation scale, characterizing positive, $(Y_{t}-\langle
Y_{t}\rangle )/t>0$, and negative, $(Y_{t}-\langle Y_{t}\rangle )/t<0$,
deviations respectively. The positive tail has the form specific for the
current distribution for $m$ free independent Bernoulli particles, while the
form of the negative tail looks like that of the distribution for one
Bernoulli particle that makes $m$ steps at a time or $m$ particles jumping
one step synchronously. Such asymmetry reveals different mechanisms of
positive and negative fluctuations. For positive deviations $m$ particles
have to be "accelerated" independently. The main contribution to the
negative deviations comes from the events when the first particle
decelerates all particles following behind. Our results are to be compared
to the ones obtained by Derrida and Lebowitz \cite{Derrida Lebowitz} (see
also \cite{Derrida Appert}) for the large deviation function of the TASEP
current on a ring. In particular, they have shown that for large $m$ the
non-universal tails of the current probability have the form $%
P(Y_{t}/t=v)\sim e^{mH_{+}(v/m)t}$ and $P(Y_{t}/t=v)\sim e^{H_{-}(v/m)t}$
for positive and negative deviations respectively. Our results have the same
functional form even for finite $m$ with $H_{-}(v)=H_{+}(v)$ being a simple
large deviation function of the Bernoulli process. Remarkably, such a
mechanism survives on the infinite lattice, despite the particles drifting
apart from each other in the course of time, so that they meet less and less
often. The acceleration-deceleration asymmetry was also observed in the
large deviations of the distance travelled by individual particles in TASEP
studied in \cite{Harris Rakos Schuetz} for a general case of particle
dependent hopping rates. There the negative large deviations do not depend
on the order number of a particle whereas the positive ones do.

The result obtained for the SVW in the transition region $(1-\kappa )\sqrt{t}%
=const$ provides us with the limiting distribution of the particle current
measured at the diffusive scale, $|Y_{t}-\langle Y_{t}\rangle |\sim \sqrt{t}$%
. The distribution obtained is parameter free, dependent only on the number
of particles. We expect that it is a universal distribution for the systems
of particles performing a driven diffusion, independent of the details of
microscopic dynamics. The distribution has a skew, non-Gaussian form, with
tails matching the large deviation behaviour asymptotically.

Several directions for future work can be mentioned. First direct
continuation of the present paper is an asymptotic study of the function $%
f_{m}\left( \alpha \right) $ as $m\rightarrow \infty $. Considering a
certain common scaling with the variable $\alpha $ is expected to give a new
universal scaling function characterizing the KPZ class. Note that in the
limiting cases this function reduces to Mehta integrals, which are the
probability normalization constants for Gaussian random matrix ensembles,
namely the Unitary and Orthogonal ensembles. The asymptotic evaluation of
these integrals performed in the random matrix theory resulted in
non-trivial densities of critical points, from which comes the leading
contribution to the integrals \cite{Mehta book}. It would be interesting to
see how these densities transform from one to another as the parameters vary
between the limiting cases corresponding to different ensembles. It would
also be interesting to study the TASEP confined in a ring. The starting
point of this analysis could be the recently obtained expressions for the
Green functions \cite{Priezzhev ring}. In this way one could obtain a
scaling function that characterizes the behaviour of KPZ interfaces in
finite systems.

Another possible development of the present article is a generalization of
the above mentioned results for the probability distributions of the
distance travelled by an individual particle in the TASEP and the
corresponding correlation functions to SVW case. Note that similar results
exist also for the VW \cite{Baik}. In both cases the appropriately rescaled
distribution of the distance travelled by a single particle starting from a
half filled lattice is shown to converge, to the so-called Tracy-Widom
distributions \cite{Tracy Widom}, which appear in the random matrix theory
as a distribution of maximal eigenvalue in the Gaussian ensembles, unitary
in case of TASEP and orthogonal for VW. The SVW model establishes a bridge
between these two cases. However, its Green function has neither a T\"{o}%
plitz form like in VW nor a special structure like in TASEP, which allowed
Sasamoto, \cite{Sasamoto}, to reinterpret it again as a problem of the VW
and finally to represent its evolution as a determinantal point process.
Therefore, a significant extension of the existing techniques is in order.
In this connection we should mention the recent advance for the Partially
Asymmetric Simple Exclusion Process \cite{Tracy Widom PASEP1}-\cite{Tracy
Widom PASEP3}, which was made only on the basis of the Bethe Ansatz solution
without use any free fermionic representation like VW.

An interesting example of SVW has been proposed recently by Johansson \cite%
{Johansson Aztec} in his analysis of a domino tilling problem on the Aztec
diamond known as the arctic circle problem. It was shown that the domino
configurations are in one-to-one correspondence with trajectories of an $n$%
-particle process which is defined as follows. At each discrete time step a
particle jumps forward with probability $q$ or stays put with probability $%
p=1-q$. If the next site is occupied, the probability to stay put is $%
1-q(1-\kappa )$ as in the SVW model. In addition, after each step, a
particle $i$ can be translated back to the distance $s_{i}$ with probability
$q^{s_{i}}$ provided that $s_{i}<X_{i}-X_{i-1}$ for all $i$. If one chooses $%
\kappa =-q$, the model belongs to the free fermion class and its transition
probabilities admit a determinant representation. It has been shown in \cite%
{Johansson Aztec} that the position of the first particle is described by
the Airy process in the thermodynamic limit for the domain wall boundary
conditions in the domino tiling problem. By similarity of the models, one
may expect that the extremal statistics of the SVW model also exhibits a
kind of Tracy-Widom distribution for appropriate initial conditions.

\section{Acknowledgements}

The authors are grateful to Gunter Sch\"utz for stimulating discussion. The
work is partially supported by the RFBR grant No. 06-01-00191a.

\end{document}